\documentclass[lettersize,journal]{IEEEtran}
\usepackage{cite}
\usepackage{amsmath,amssymb,amsfonts}
\usepackage{algorithmic}
\usepackage{algorithm}
\usepackage{array}
\usepackage[caption=false,font=normalsize,labelfont=sf,textfont=sf]{subfig}
\usepackage{textcomp}
\usepackage{stfloats} 
\usepackage{graphicx}
\usepackage{bm}
\usepackage{cases}
\usepackage{multirow}
\usepackage{xcolor}
\usepackage{url}
\usepackage{verbatim}
\newtheorem{theorem}{\bf{\textit{Theorem}}}[section]
\newtheorem{lemma}{\bf{\textit{Lemma}}}[section]
\usepackage{graphicx}
\usepackage{cite}
\usepackage{hyperref}
\begin{document}

\title{Fairness-Aware Beamforming for Polarimetric ISAC Systems with Polarization-Reconfigurable Antennas}

\author{Weijie Xiong, Jingran Lin, Di Jiang, Cunhua Pan, Hongli Liu, Kai Zhong, and Qiang Li
\thanks{This work was supported by the National Natural Science Foundation of China (No. 62501112) and in part by the China Postdoctoral Science Foundation (No. 2025M773511). \textit{(Corresponding author: Jingran Lin)}.}
\thanks{Jingran Lin and Qiang Li are with the School of Information and Communication Engineering, University of Electronic Science and Technology of China, Chengdu 611731, China, the Laboratory of Electromagnetic Space Cognition and Intelligent Control, Beijing 100083, China, and also with the Tianfu Jiangxi Laboratory, Chengdu, Sichuan 641419, China (e-mail: jingranlin@uestc.edu.cn; lq@uestc.edu.cn).}
\thanks{Weijie Xiong, Di Jiang, Hongli Liu, Kai Zhong are with the School of Information and Communication Engineering, University of Electronic Science and Technology of China, Chengdu 611731, China (e-mail: 202311012313@std.uestc.edu.cn; dijiang@uestc.edu.cn; hongliliu@std.uestc.edu.cn; 201921011206@std.uestc.edu.cn).}
\thanks{Cunhua Pan is with the National Mobile Communications Research Laboratory, Southeast University, China (e-mail: cpan@seu.edu.cn).}
}

\markboth{Journal of \LaTeX\ Class Files,~Vol.~14, No.~8, August~2021}%
{Shell \MakeLowercase{\textit{et al.}}: A Sample Article Using IEEEtran.cls for IEEE Journals}


\maketitle

\begin{abstract}
Polarization diversity offers significant flexibility for enhancing integrated sensing and communications (ISAC). However, conventional dual-polarized arrays typically require dedicated radio-frequency (RF) chains for each polarization branch, leading to prohibitive hardware costs. To address this, polarization-reconfigurable (PR) antennas have emerged as a cost-effective alternative, enabling polarization flexibility with reduced hardware complexity by driving two polarization branches with a single RF chain. In this paper, we investigate fairness-aware beamforming for ISAC systems equipped with PR antennas. Specifically, we jointly optimize the transmit beamforming and PR control coefficients to maximize the minimum signal-to-interference-plus-noise ratio (SINR) for communication users and the minimum signal-to-clutter-plus-noise ratio (SCNR) for sensing targets. The resulting problem is highly nonconvex and nonsmooth due to the strong coupling among optimization variables in the max–min objective, as well as the nonconvex spherical constraints imposed by the PR antennas. To tackle this, we derive an equivalent smooth reformulation by introducing auxiliary variables and transforming the minimum operators into inequality constraints. Subsequently, we develop an exact-penalty product Riemannian manifold gradient descent (EP-PRMGD) algorithm, which integrates an exact penalty method with Riemannian optimization to guarantee convergence to a Karush–Kuhn–Tucker (KKT) point. Numerical results demonstrate that the proposed PR-enabled ISAC scheme achieves performance comparable to dual-polarized architectures while utilizing only half the RF chains, thereby validating its effectiveness in balancing fairness and hardware efficiency.
\end{abstract}

\begin{IEEEkeywords}
Integrated sensing and communications (ISAC), polarization-reconfigurable (PR) antennas, fairness-aware beamforming, max-min optimization, Riemannian manifold.
\end{IEEEkeywords}

\section{Introduction}
The evolution of wireless communication systems toward sixth-generation (6G) networks and beyond has highlighted the importance of integrated sensing and communications (ISAC), which unifies radar sensing and communication functionalities within a shared spectrum and hardware platform \cite{ghosh2025unified}. This paradigm responds to the growing demand for spectral efficiency, cost-effective hardware implementations, and support for emerging applications such as autonomous vehicles, smart cities, and augmented reality \cite{pan2020multicell,gong2024toward}. By enabling concurrent sensing and communication, ISAC offers substantial benefits, including improved resource utilization, reduced hardware complexity, and enhanced system scalability \cite{liu2025fundamental,pan2021reconfigurable}. However, fully realizing these gains relies on efficient resource allocation, since competition for shared resources introduces critical tradeoffs between sensing and communication \cite{dong2022sensing}.

Among various resource allocation strategies, transmit beamforming plays a central role in ISAC systems \cite{li2024transmit}. However, conventional transmit beamforming designs often optimize aggregate objectives, such as sum-rate maximization \cite{zhang2024efficient} or total mean-square-error minimization \cite{wang2024unified}, which inherently favor strong links and targets. Consequently, resources are biased toward advantageous channels, leaving cell-edge users and weak targets unable to meet their basic requirements. To alleviate these drawbacks, fairness-aware beamforming has been widely investigated to improve the performance of disadvantaged nodes and enforce a quality-of-service floor \cite{dou2024integrated,sobhi2024joint,ji2023robust,shakoor2025max,zhang2024max}. Specifically, \cite{dou2024integrated} maximized the minimum user signal-to-interference-plus-noise ratio (SINR) subject to sensing constraints, while \cite{sobhi2024joint} proposed a fairness-aware hybrid beamforming scheme for millimeter-wave systems. To account for channel uncertainties, \cite{ji2023robust} developed robust max-min beamforming designs under worst-case channel state information (CSI) errors. Furthermore, \cite{shakoor2025max,zhang2024max} exploited intelligent reflecting surfaces to enhance cell-edge signal strength and balance the trade-off between communication rate and sensing accuracy under fairness constraints. However, strict fairness requirements are more demanding than aggregate objectives because they must ensure worst-case user and target performance rather than optimize an average or sum utility. As a result, satisfying worst-case constraints can quickly exhaust the spatial degrees of freedom (DoFs), and further reallocation may incur a noticeable loss in overall system capacity \cite{shakoor2025max}.

To alleviate these drawbacks, polarization diversity has been recognized as an effective ways to enrich the design space for resource allocation. Crucially, it provides controllable degrees of freedom (DoFs) in the polarization domain without increasing the physical array aperture \cite{xia2025ris}. By exploiting the orthogonality between polarization components (e.g., horizontal and vertical), this approach introduces an additional dimension for signal processing. This capability allows the system to combat communication fading and and reveal target features that cannot be effectively captured by conventional sensing schemes \cite{xia2025ris}. Substantial research has validated the benefits of polarization in both communication and sensing applications \cite{wu2023physical,yang2022dual,shahzadi2024dual,yin2021clutter,xia2025ris,zhang2024dual}. In wireless communications, \cite{wu2023physical,yang2022dual} showed that dual-polarized MIMO architectures can reduce outage probability and improve channel capacity in multipath environments. In radar sensing, \cite{shahzadi2024dual,yin2021clutter} utilized polarimetric scattering information to improve target parameter estimation accuracy and enhance detection performance in clutter-rich environments. More recently, \cite{xia2025ris,zhang2024dual} demonstrated that dual-polarized arrays can jointly improve user communications and target localization in ISAC systems.

Despite these performance gains, fully exploiting polarization diversity typically requires dual-polarized antennas with dedicated radio frequency (RF) chains for each polarization branch. This architecture essentially doubles the number of RF chains and analog-to-digital converters (ADCs) compared to conventional systems \cite{castellanos2023linear}. As antenna arrays scale up to meet the strict performance requirements of 6G, such hardware redundancy leads to prohibitive costs and circuit complexity, limiting their practical deployment \cite{lee2025integrated}.

To overcome this bottleneck, polarization-reconfigurable (PR) antennas have emerged as a cost-effective alternative. Unlike fixed dual-polarized antennas, a PR antenna element is driven by a single RF chain but can dynamically adjust its polarization state through tunable analog components, thereby reducing the RF hardware cost by half while retaining polarization flexibility \cite{castellanos2023linear}. This technology has been successfully applied in various wireless scenarios to enhance system performance without increasing hardware complexity \cite{castellanos2023linear, castellanos2021mimo, zhou2024polarforming, shao2025polarized, lee2025integrated}. For instance, \cite{castellanos2023linear,castellanos2021mimo} demonstrated that PR arrays can substantially enhance channel capacity in MIMO communications while maintaining a single RF chain per antenna by aligning signal polarization with the channel eigenmodes. From a theoretical perspective, \cite{zhou2024polarforming} established a polarforming framework to systematically model PR-enabled channels, quantifying the performance gains achievable with reduced hardware overhead. Furthermore, \cite{shao2025polarized} explored polarized movable antennas to optimize wireless channels, while \cite{lee2025integrated} recently utilized PR arrays to achieve promising gains in target detection without the cost of full dual-polarized architectures.

While PR antennas have demonstrated significant potential in ISAC systems, existing studies have predominantly focused on aggregate performance metrics. Consequently, the application of PR technology to address the critical fairness issue, specifically, guaranteeing service quality for worst-case users and targets, remains unexplored. This gap is significant because the additional DoFs offered by PR antennas provide a hardware-efficient mechanism to compensate for the intense resource competition in max-min fairness designs. To bridge this gap, this paper proposes a novel fairness-aware beamforming framework for PR-assisted ISAC systems, aiming to achieve a favorable trade-off between worst-case performance and hardware complexity. The main contributions of this work are summarized as follows:
\begin{itemize}
\item Unlike existing works \cite{castellanos2023linear, castellanos2021mimo, zhou2024polarforming, shao2025polarized, lee2025integrated} that primarily focus on aggregate performance metrics in ISAC systems, we propose a fairness-aware beamforming framework equipped with PR antennas. By exploiting the polarization DoFs, our design aims to guarantee service equity for the worst-case communication users and sensing targets, while achieving a hardware-efficient architecture that requires only half the RF chains compared to conventional dual-polarized systems.

\item We formulate the fairness-aware design as a highly non-convex and non-smooth max-min optimization problem. Specifically, we seek to maximize a weighted combination of the minimum user SINR and the minimum target SCNR. This is achieved by jointly optimizing the transmit beamforming matrix, the polarization combining matrices, and the receive filter, subject to the total transmit power budget the spherical constraints of PR antennas.

\item To solve the formulated problem, instead of adopting conventional alternating optimization (AO) approaches \cite{lee2025integrated,zhang2024dual}, which artificially decouple the variables and generally lack strict convergence guarantees \cite{shi2020penalty}, we develop a unified exact penalty product Riemannian manifold gradient descent (EP-PRMGD) algorithm. The max-min objective is first smoothed via an epigraph reformulation, and the resulting inequality constraints are strictly enforced through an exact penalty (EP) method. We then show that these nonconvex constraints jointly characterize a product Riemannian manifold. By incorporating the constraints into the problem geometry, the original formulation is transformed into an unconstrained optimization problem over a unified manifold. This enables the simultaneous update of all highly coupled variables and admits a theoretical guarantee of convergence to a Karush-Kuhn-Tucker (KKT) point.

\end{itemize}

The remainder of this paper is organized as follows. Section II introduces the system model. Section III details the problem reformulation, the proposed EP-PRMGD algorithm, and the corresponding convergence analysis. Section IV presents the numerical results, and Section V concludes the paper.

The following notations are used throughout this paper. Vectors and matrices are denoted by bold lowercase and bold uppercase letters, respectively. $(\cdot)^T$, $(\cdot)^H$, and $(\cdot)^*$ denote the transpose, Hermitian transpose, and complex conjugate. $\mathbf{I}_N$ denotes the $N\times N$ identity matrix. $\mathbb{C}^{M\times N}$ and $\mathbb{R}^{M\times N}$ denote the sets of $M\times N$ complex- and real-valued matrices, respectively. $\mathcal{CN}(\boldsymbol{\mu},\mathbf{\Sigma})$ denotes the circularly symmetric complex Gaussian distribution with mean $\boldsymbol{\mu}$ and covariance $\mathbf{\Sigma}$, and $\mathbb{E}\{\cdot\}$ denotes expectation. $\mathrm{Tr}(\cdot)$, $\|\cdot\|_F$, $\|\cdot\|_2$, and $|\cdot|$ denote the trace, Frobenius norm, Euclidean norm, and absolute value, respectively. $\mathrm{diag}(\mathbf{a})$ forms a diagonal matrix with $\mathbf{a}$ on its main diagonal, and $\mathrm{blkdiag}(\cdot)$ forms a block-diagonal matrix from its arguments. $\Re\{\cdot\}$ denote the (element-wise) real-part operators and $\otimes$ denotes the Kronecker product.

\section{System Model and Problem Formulation}

\begin{figure*}[!t]
  \begin{center}
  \includegraphics[ width=0.85\textwidth]{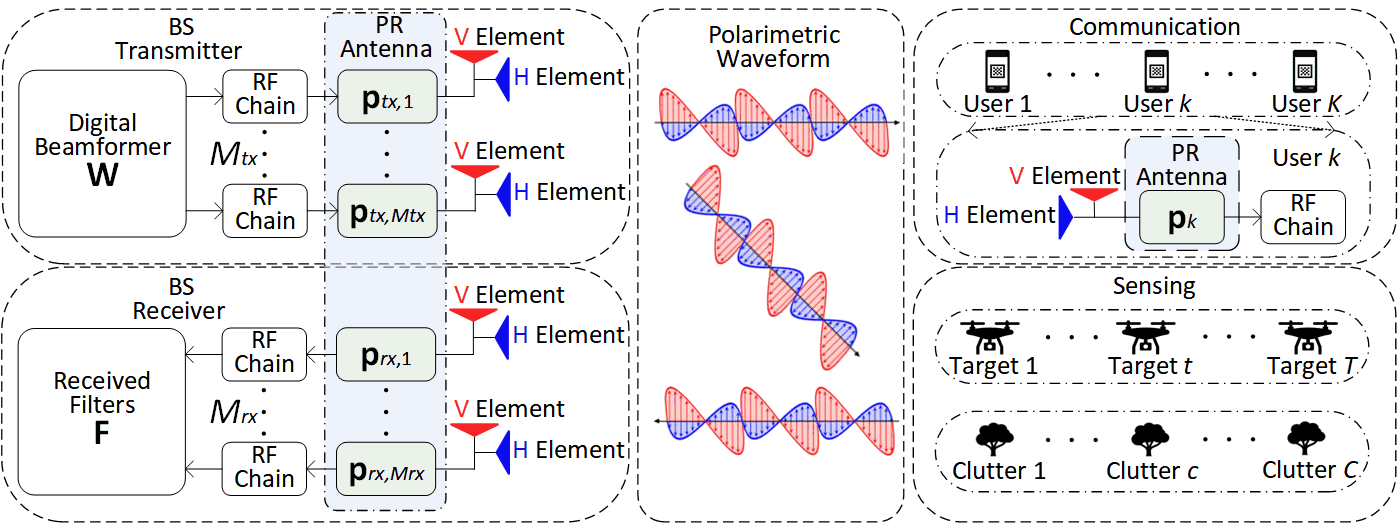}\\
  \caption{The polarimetric ISAC system with PR antennas.}\label{systemmodel}
  \end{center}
   \vspace{-10pt}
\end{figure*}

Consider a polarimetric ISAC system, as depicted in Fig. \ref{systemmodel}, where a dual-functional base station (BS) simultaneously serves $K$ single-antenna communication users and performs sensing of $T$ point-like targets in the presence of $C$ point-like clutter sources. The transmitting BS employs uniform linear arrays (ULAs) with $M_{tx}$ transmit antennas and $M_{rx}$ receive antennas. A key feature is that both the BS and the users are equipped with PR antennas. Each PR antenna employs a single RF chain to drive two orthogonally polarized elements, horizontal (H) and vertical (V), thereby enabling polarimetric operation with low hardware cost and enhancing the flexibility and performance of both communication and sensing \cite{castellanos2023linear}.

\subsection{Transmit Signal Model}
At baseband, the BS generates two independent sets of data streams: a communication stream $\mathbf{s}_c = [s_1, \dots, s_K]^T \in \mathbb{C}^{K}$ intended for the $K$ users, and a dedicated sensing stream $\mathbf{s}_r = [s_{K+1}, \dots, s_{K+L_r}]^T \in \mathbb{C}^{L_r}$ comprising $L_r$ independent radar waveforms. The entries of each stream are independent and identically distributed (i.i.d.), satisfying $\mathbb{E}\{{\bf s}_c {\bf s}_c^H\} = {\bf I}_K$ and $\mathbb{E}\{{\bf s}_r {\bf s}_r^H\} = {\bf I}_{L_r}$, while ${\bf s}_c$ and ${\bf s}_r$ are mutually independent. The streams are processed by their respective digital beamforming matrices ${\bf W}_c =[{\bf w}_1,...,{\bf w}_K]\in \mathbb{C}^{M_{tx} \times K}$ and ${\bf W}_r =[{\bf w }_{K+1},...,{\bf w}_{K+L_r}] \in \mathbb{C}^{M_{tx} \times L_r}$. The resultant signal at the RF-chain inputs is,
\begin{equation}
    {\bf \hat x} = {\bf W}_c {\bf s}_c + {\bf W}_r {\bf s}_r = {\bf W}{\bf s} \in \mathbb{C}^{M_{tx}},
\end{equation}
where the overall beamformer is ${\bf W} = [{\bf W}_c, {\bf W}_r] \in \mathbb{C}^{M_{tx} \times (K+L_r)}$ and the transmit symbol vector is ${\bf s} = [{\bf s}_c^T, {\bf s}_r^T]^T \in \mathbb{C}^{K+L_r}$.

The vector ${\bf \hat x}$ is then applied to $M_{tx}$ RF chains and fed to the PR transmit antennas. Each antenna employs a single-port architecture, in which one RF chain drives two orthogonally polarized elements: H and V. Polarization reconfigurability is realized by controlling the relative amplitudes of the signals delivered to these two branches. Such a mechanism is practically implementable using mature analog RF circuitry, e.g., an analog combiner based on a varactor-controlled network \cite{gao2006polarization,cai2016continuously} or variable attenuators \cite{zhou2024polarforming}, which confirms the practical feasibility of the adopted PR antenna architecture. Unlike conventional dual-polarized ISAC systems that require separate RF chains for the H and V branches, the proposed architecture enables polarization reconfiguration with a single shared RF chain, thereby significantly reducing hardware cost \cite{chen2025hybrid}. Specifically, the polarization control can be modeled by a polarization combining matrix, defined as \cite{castellanos2023linear},
\begin{equation}
    {\bf P}_{tx} = \text{blkdiag}({\bf p}_{tx,1}, \dots, {\bf p}_{tx,M_{tx}}) \in \mathbb{R}^{2M_{tx} \times M_{tx}}, \label{commar1}
\end{equation}
where ${\bf p}_{tx,m_{tx}} \in \mathbb{R}^{2 }$ is the real-valued combining vector for the $m_{tx}$-th antenna, with $m_{tx} \in [1,M_{tx}]$. Specifically, ${\bf p}_{tx,m_{tx}}=[p^{H}_{tx,m_{tx}}, p^{V}_{tx,m_{tx}}]^T$, with $p^{H}_{tx,m_{tx}} \in  \mathbb{R}$ and $p^{V}_{tx,m_{tx}} \in  \mathbb{R}$ denote the weighting coefficients for the H and V polarization components, respectively. It sets the relative power split between H and V and satisfies $\|{\bf p}_{tx,m_{tx}}\| = \sqrt{(p^{H}_{tx,m_{tx}})^2 + (p^{V}_{tx,m_{tx}})^2} = 1$.

Combining these stages, the dual-polarized transmit signal from the BS is expressed as,
\begin{equation}
    {\bf x} = {\bf P}_{tx} {\bf \hat x} = {\bf P}_{tx} {\bf W}{\bf s} \in \mathbb{C}^{2M_{tx}}.
\end{equation}

\subsection{Received Signal Model}
The BS is assumed to have perfect CSI, which can be obtained using standard estimation methods \cite{lauinger2022blind,hu2025channel,wen2022compressive}. Consequently, the dual-polarized downlink channel from the BS to user $k$ is denoted by ${\bf H}_k\in\mathbb{C}^{2\times 2M_{tx}}$, for $k\in\{1, \dots, K\}$. Each user $k$ is equipped with a single PR antenna characterized by a polarization combiner ${\bf p}_k=[p^{H}_{k}, p^{V}_{k}]^T\in\mathbb{R}^{2}$. Then, the received signal at user $k$ is therefore,
\begin{equation}
y_k = {\bf p}_k^T{\bf H}_k{\bf P}_{tx}
\left( {\bf w}_k s_k+ \sum_{j\neq k}^{K+L_r} {\bf w}_j s_j \right) + n_k,
\end{equation}
where $n_k\sim\mathcal{CN}(0,\sigma_k^2)$ is the circularly-symmetric complex Gaussian (CSCG) noise \cite{pan2020intelligent}.

For monostatic sensing, the dual-polarized channel response of target $t \in \{1, \dots, T\}$ is denoted by ${\bf G}_t\in\mathbb{C}^{2M_{rx}\times 2M_{tx}}$, and that of clutter source $i \in \{T+1, \dots, T+C\}$ by ${\bf G}_i\in\mathbb{C}^{2M_{rx}\times 2M_{tx}}$. The BS receiver is equipped with $M_{rx}$ PR antennas, represented by a polarization-combining matrix ${\bf P}_{rx} = \text{blkdiag}({\bf p}_{rx,1}, \dots, {\bf p}_{rx,M_{rx}})\in\mathbb{R}^{2M_{rx}\times M_{rx}}$ defined similarly to (\ref{commar1}), where each ${\bf p}_{rx,m_{rx}}=[p^{H}_{rx,m_{rx}}, p^{V}_{rx,m_{rx}}]^T\in\mathbb{R}^{2}$. The received echo signal for sensing is then,
\begin{equation}
{\bf y}_r = {\bf P}_{rx}^T
\left( \sum_{t=1}^{T} {\bf G}_t + \sum_{i=T+1}^{T+C} {\bf G}_i \right)
{\bf P}_{tx}{\bf W}{\bf s} + {\bf n}_r,
\end{equation}
where ${\bf n}_r\sim\mathcal{CN}({\bf 0},\sigma_r^2{\bf I}_{M_{rx}})$ is CSCG noise at the receiver.

\subsection{Channel Model}
We assume a far-field propagation environment, in which the array aperture is small relative to the propagation distance \cite{bhagavatula2010new}. Under this assumption, the incident wavefront of each path can be approximated as planar across the array. The path’s complex gain and polarization transformation are identical for all antenna elements, while the array response is characterized by element-dependent phase shifts determined by the array geometry and the path’s angle of departure or arrival \cite{castellanos2023linear}.

Following \cite{bhagavatula2010new,lee2025integrated}, we adopt a classical multipath model for the communication link. Since the antenna pattern is decomposed into two orthogonal linear polarization components, the field response matrix (FRM) of the dual-polarized channel for the $l$-th path for user $k$ is,
\begin{equation}
{\bf A}(\theta_{k,l}) =  {\bf a}(\theta_{k,l}) \otimes {\bf I}_2 \in \mathbb{C}^{2M_{tx} \times 2},
\end{equation}
where ${\bf a}(\theta_{k,l}) \in \mathbb{C}^{M_{tx}}$ is the transmit steering vector given as,
\begin{equation}
{\bf a}(\theta_{k,l})
= [1,e^{-j\frac{2\pi d}{\lambda}\sin(\theta_{k,l})}\ldots e^{-j\frac{2\pi d}{\lambda}(M_{tx}-1)\sin(\theta_{k,l})}]^T, \label{steeringv}
\end{equation}
where $d$ is the inter-element spacing, $\lambda$ is the wavelength, and $\theta_{k,l}$ denotes the AoD of the $l$-th path of user $k$. Assume that the channel from the BS to user $k$ contains $L_k$ paths. The communication channel is then expressed as \cite{lee2025integrated},
\begin{equation}
{\bf H}_k = \frac{1}{\sqrt{L_k}}\sum_{l=1}^{L_k} \beta_{k,l}{\bf J}_{k,l} {\bf A}^{T}(\theta_{k,l}), \label{define2}
\end{equation}
where $\beta_{k,l}\in\mathbb{C}$ is the complex path gain, and ${\bf J}_{k,l}\in \mathbb{C}^{2 \times 2}$ denotes the depolarization matrix that models scattering-induced polarization coupling, given as \cite{lee2025integrated},
\begin{equation}
    {\bf J}_{k,l} =\frac{1}{\sqrt{1 +\chi_{k,l} }}\begin{bmatrix} e^{j \alpha_{k,l}^{HH} } & \sqrt{\chi_{k,l}}e^{j \alpha_{k,l}^{HV} } \\ \sqrt{\chi_{k,l}}e^{j \alpha_{k,l}^{VH} } & e^{j \alpha_{k,l}^{VV} } \end{bmatrix}, \label{leakege}
\end{equation}
where $\chi_{k,l}\ge0$ is the cross-polarization discrimination (XPD) factor that quantifies cross-polar energy leakage during propagation \cite{song2015adaptive}, and ${ \alpha_{k,l}^{AB} (A,B\in\{H,V\})}$ denote the phase shifts from polarization $A$ to $B$.

Without loss of generality, a single-path model is adopted for sensing, since multi-bounce reflections typically suffer severe attenuation. We therefore assume that only the line-of-sight (LoS) path from the BS transmitter to the target/clutter and the return path from the target/clutter to the BS receiver are detectable. Let $q\in\{t,i\}$ denote either the $t$-th target or the $i$-th clutter. The FRM from the BS to $q$ and the FRM from $q$ to the BS are respectively given as,
\begin{equation}
\begin{aligned}
& {\bf A}_{tx}(\theta_{q}) =  {\bf a}_{tx}(\theta_{q}) \otimes {\bf I}_2   \in \mathbb{C}^{2M_{tx} \times 2},\\
& {\bf A}_{rx}(\theta_{q}) = {\bf a}_{rx}(\theta_{q}) \otimes {\bf I}_2   \in \mathbb{C}^{2M_{rx} \times 2},
\end{aligned}
\end{equation}
where ${\bf a}_{tx}(\theta_{q})\in \mathbb{C}^{M_{tx}}$ and ${\bf a}_{rx}(\theta_{q})\in \mathbb{C}^{M_{rx}}$ are steering vectors defined as in \eqref{steeringv}. Hence, the monostatic sensing channel from the transmitter to $q$ and back to the receiver is,
\begin{equation}
{\bf G}_{q} = \beta_{q} {\bf A}_{rx}(\theta_{q}){\bf J}_{q}{\bf A}^{T}_{tx}(\theta_{q}),
\end{equation}
where $\beta_{q}\in\mathbb{C}$ denotes the complex two-way path gain, and ${\bf J}_{q}\in \mathbb{C}^{2 \times 2}$ is the depolarization matrix, defined same to \eqref{leakege}.

\subsection{Performance Metrics}
To evaluate the communication and sensing performance, the SINR and signal-to-clutter-plus-noise ratio (SCNR) are adopted as the respective performance metrics. Based on the above discussion, the SINR for user $k$ is expressed as
\begin{equation}
\text{SINR}_k = \frac{| {\bf p}_k^T {\bf  H}_k {\bf P}_{tx} {\bf w}_k |^2}
{\sum_{j \neq k}^{K+L_r} | {\bf p}_k^T {\bf  H}_k {\bf P}_{tx} {\bf w}_j |^2 + \sigma_k^2}.
\end{equation}

For sensing, let ${\bf F}=[{\bf f}_1,\ldots,{\bf f}_T] \in \mathbb{C}^{M_{rx} \times T}$ denote the collection of linear filters ${\bf f}_t \in \mathbb{C}^{M_{rx}}$ applied at the receiver side of the BS to maximize the output SCNR. The combined signal for target $t$ is given by,
\begin{equation}
r_t = {\bf f}_t^H{\bf y}_r = {\bf f}_t^H{\bf P}_{rx}^T ({\bf  G}_t + \sum_{q \neq t}^{T+C} {\bf  G}_q ) {\bf P}_{tx}{\bf W}{\bf s} + {\bf f}_t^H {\bf n}_r,
\label{revicesignal}
\end{equation}
Hence, the SCNR for target $t$ is expressed as,
\begin{equation}
    \text{SCNR}_t =\frac{\|{\bf f}_t^H{\bf P}_{rx}^T {\bf  G}_t {\bf P}_{tx}{\bf W}\|_2^2}
{  \sum_{q \neq t}^{T+C} \|{\bf f}_t^H {\bf P}_{rx}^T {\bf  G}_q {\bf P}_{tx}{\bf W} \|_2^2 +  \sigma_r^2 \|{\bf f}_t\|_2^2 } .
\end{equation}

\subsection{Problem Formulation}
Our objective is to jointly optimize the transmit beamforming matrix $\bf W$, the polarization combining matrices ${\bf P}_{tx}$ and ${\bf P}_{rx}$, the per-user combiners $\{{\bf p}_k\}_{k=1}^K$, and the receive filter matrix $\bf F$ to achieve a balanced trade-off between communication and sensing performance. Specifically, we pursue a design that guarantees fairness among communication users and sensing targets. To this end, we adopt the worst-case (minimum) SINR and SCNR as the respective performance metrics. The joint design problem is formulated as,
\begin{subequations}
\begin{align}
\max_{ \left\{\begin{subarray}{c}
{\bf W},{\bf P}_{tx},{\bf P}_{rx},\\
\{{\bf p}_k\}_{k=1}^K,{\bf F}
\end{subarray} \right\} }&
\;(1-\rho)\min_{k\in[1,K]} \text{SINR}_k
+\rho\min_{t\in[1,T]}\text{SCNR}_t, \label{objori}\\
\text{s.t.}\quad & \text{Tr}({\bf W}{\bf W}^H) =P,\label{transmitpowcon}\\ 
&\|{\bf p}_{tx,{m_{tx}}}\|_2 = 1,\quad \forall m_{tx} \in [1,M_{tx}],\label{conpo1}\\
&\|{\bf p}_{rx,{m_{rx}}}\|_2 = 1,\quad  \forall m_{rx} \in [1,M_{rx}],\label{conpo2}\\
& \|{\bf p}_k\|_2 =1,\quad  \forall k \in [1,K],\label{conpo3}
\end{align}
\label{proFor}%
\end{subequations}
where \eqref{transmitpowcon} is the transmit power constraint; \eqref{conpo1}–\eqref{conpo3} are spherical constraints on the polarization vectors. Here, $P$ denotes the total transmit power, and the parameter $\rho \in [0,1]$ controls the trade-off between communication and sensing performance. Generally, problem \eqref{proFor} is challenging for the following reasons: (i) the constraints \eqref{transmitpowcon}–\eqref{conpo3} are nonconvex; (ii) the pointwise-minimum objective \eqref{objori} is nonsmooth; (iii) the SINR/SCNR terms lead to a fractional, nonconvex objective; and (iv) the decision variables are highly coupled within the objective.

In general, problem (\ref{proFor}) is inherently non-convex and non-smooth. While the AO method is widely used for such coupled ISAC designs \cite{lee2025integrated,zhang2024dual}, its artificial block-wise decoupling ignores the strong interplay among variables. Consequently, AO often gets trapped in the limit points and lacks strict convergence guarantees for the original problem \cite{shi2020penalty}. To overcome these limitations, we propose a unified framework that jointly optimizes all variables. First, we smooth the max–min objective by transforming the minimum terms into inequality constraints, which are then handled via an exact penalty (EP) approach. Furthermore, instead of treating the non-convex constraints \eqref{transmitpowcon}–\eqref{conpo3} separately, we recognize that they collectively constitute a product Riemannian manifold. By integrating all highly coupled variables into a single point on this curved space, we develop the EP-PRMGD algorithm to update them simultaneously. This direct and joint optimization completely preserves the structure of the problem, thereby ensuring a theoretical guarantee of convergence to a Karush-Kuhn-Tucker (KKT) point \cite{boyd2004convex}. The details are derived below.

\section{Problem Reformulation and Proposed Method}

\subsection{Problem Reformulation}
To begin with, by applying the epigraph reformulation \cite{boyd2004convex}, problem \eqref{proFor} can be equivalently written as the following smooth optimization problem,
\begin{subequations}
\begin{align}
\max_{ {\bf W},{\bf P}_{tx},{\bf P}_{rx},
\{{\bf p}_k\}_{k=1}^K,{\bf F},a,b}&(1-\rho)a+\rho b, \\
\text{s.t.}\quad  & a \le \text{SINR}_k, \quad\forall k \in [1,K], \label{reformuatecons1}\\ 
& b \le \text{SCNR}_t, \quad\forall t \in [1,T], \label{reformuatecons2}\\ 
& a\ge 0,\quad b \ge 0,\label{reformuatecons3}\\
& \text{\eqref{transmitpowcon}-\eqref{conpo3} are satisfied.}\notag
\end{align}
\label{reFor1}%
\end{subequations}
Equivalently, by transforming the maximization into a minimization problem, we have,
\begin{equation}
    \begin{aligned}
  \min_{ {\bf W},{\bf P}_{tx},{\bf P}_{rx},
\{{\bf p}_k\}_{k=1}^K,{\bf F},a,b}&(\rho-1)a-\rho b, \\     
\text{s.t.}\quad \text{\eqref{transmitpowcon}-\eqref{conpo3}},&   \text{\;\eqref{reformuatecons1}-\eqref{reformuatecons3} are satisfied.}\label{reform2euq}
    \end{aligned}
\end{equation}
The constraints in \eqref{transmitpowcon}--\eqref{conpo3} naturally correspond to spherical manifolds in the Riemannian space \cite{lee2018introduction}, which motivates a manifold-based treatment of \eqref{reform2euq}. However, although the epigraph reformulation removes the non-smoothness of the original max-min objective, it also introduces the inequality constraints in \eqref{reformuatecons1}--\eqref{reformuatecons3}, which lack a smooth manifold structure and thus preclude the direct application of manifold-based algorithms. To address this issue, we incorporate these inequalities into the objective via penalty terms that discourage constraint violations. Based on these observations, we propose the EP-PRMGD algorithm to efficiently solve \eqref{reform2euq}.

\subsection{Proposed EP-PRMGD Algorithm}
In this section, we propose the EP-PRMGD algorithm to solve problem \eqref{reform2euq}. We begin by applying the EP method \cite{boyd2004convex} to relax the inequality constraints and incorporate them as penalty terms within the objective function. We then construct a smooth product Riemannian manifold (PRM), which consolidates information from the individual Riemannian manifolds associated with each constraint, transforming the problem into an unconstrained optimization problem in Riemannian space. Finally, we develop the efficient PRMGD algorithm to solve this problem within the Riemannian framework.

\subsubsection{EP-Riemannian Transformation}
Directly incorporating inequality constraints into the objective does not guarantee feasibility and may induce constraint violations or numerical divergence. The EP method provides a more reliable mechanism by augmenting the objective with penalties proportional to the degree of infeasibility \cite{boyd2004convex}. Hence, \eqref{reform2euq} can be recast as an exterior-penalty problem with manifold-related constraints \eqref{transmitpowcon}–\eqref{conpo3} as follows,
\begin{equation}
\begin{aligned}
\min _{ \left\{\begin{subarray}{c}
{\bf W},{\bf P}_{tx},{\bf P}_{rx},\\
\{{\bf p}_k\}_{k=1}^K,\\{\bf F},a,b
\end{subarray} \right\} }  & \left\{\begin{array}{l} (\rho-1)a-\rho b  \\+\lambda  \textstyle\sum\limits_{k=1}^{K}   
 \max(0,{a-\text{SINR}_k}) \\+ \lambda  \textstyle\sum\limits_{t=1}^{T}   
 \max(0,{b-\text{SCNR}_t})\\+ \lambda \max(0,-a)+\lambda \max(0,-b)  \end{array}\right\}, \\
 \text { s.t. } & \text{\eqref{transmitpowcon}-\eqref{conpo3} are satisfied,}\label{reforfirsm1}
\end{aligned}
\end{equation}
where $\lambda > 0$ is the penalty factor. Note that the penalty terms $\max(0,\cdot)$ are nonsmooth and can be challenging to optimize directly. Fortunately, each is a two-term maximum that can be smoothed using the log-sum-exponential (LSE) approximation Lemma \ref{lemma1}.

\begin{lemma} \label{lemma1} \itshape
 Given $m,n \in \mathbb{R}$, it holds for any $\mu > 0$ that \text{\cite{nesterov2013introductory}},
\begin{equation}
    \begin{aligned}
        \max \{ m,n\} \le \mu \log ( \exp({\frac{m}{\mu}})+\exp({\frac{n}{\mu}})).
    \end{aligned}
\end{equation}
Moreover, the inequalities become tight as $\mu  \to  0$.
\end{lemma}

Based on this observation, \eqref{reforfirsm1} can be reformulated as a smooth EP optimization problem given as,
\begin{equation}
\begin{aligned}
\min _{ \left\{\begin{subarray}{c}
{\bf W},{\bf P}_{tx},{\bf P}_{rx},\\
\{{\bf p}_k\}_{k=1}^K,\\{\bf F},a,b
\end{subarray} \right\} }  & \left\{\begin{array}{l} (\rho-1)a-\rho b  \\+\lambda \mu \textstyle\sum\limits_{k=1}^{K}   
 \log(1+e^{(a-\text{SINR}_k)/\mu}) \\+ \lambda \mu \textstyle\sum\limits_{t=1}^{T}   \log(1+e^{(b-\text{SCNR}_t)/\mu})
 \\+ \lambda \mu \log(1+e^{-a/\mu})+\lambda \mu \log(1+e^{-b/\mu}) \end{array}\right\}, \\
 \text { s.t. } & \text{\eqref{transmitpowcon}-\eqref{conpo3} are satisfied,}\label{reforfirsm2}
\end{aligned}
\end{equation}
The updates of \(\lambda\) and \(\mu\) will be discussed in Section III.B.(3). In the \(j\)-th iteration, we first outline the process for solving the EP problem (\ref{reforfirsm2}) with a fixed penalty parameter \(\lambda^j\) and a smooth variable \(\mu^j\). Note that problem \eqref{reforfirsm2} contains the non-convex spherical constraints \eqref{transmitpowcon}–\eqref{conpo3}, which are difficult to represent and manipulate in a linear space. Manifolds, however, provide a natural framework for expressing such constraints, as they flexibly capture the underlying nonlinear structure. Their locally Euclidean geometry enables the formulation and handling of these constraints in a manner consistent with the intrinsic geometry of the problem \cite{boumal2023introduction}.

\subsubsection{Construction of the PRM}
Generally, the constraints \eqref{transmitpowcon}–\eqref{conpo3} define feasible sets in Euclidean space that are smooth submanifolds. Equipped with the induced Riemannian metric, these sets become Riemannian manifolds, which can be interpreted as restricted search domains for the optimization problem. Such manifold structures have been extensively exploited in waveform design; see \cite{xiong2025enhancing,zhong2024p2c2m}. In particular, constraint \eqref{transmitpowcon} defines the complex spherical manifold ${\cal M}_{\bf W}$,
\begin{equation}
\mathcal{M}_{\bf W}=\left\{{\bf W} \in \mathbb{C}^{M_{tx} \times (K+L_r)} | \,\|{\bf W}\|_F =\sqrt{P}\right\}. \label{manifoldW}
\end{equation}
Constraints \eqref{conpo1}–\eqref{conpo3} define spherical manifolds for the polarization vectors ${\cal M}_{{\bf p}_{tx,m_{tx}}},\forall m_{tx} \in [1,M_{tx}]$, ${\cal M}_{{\bf p}_{rx,m_{rx}}},\forall m_{rx} \in [1,M_{rx}]$, ${\cal M}_{{\bf p}_{k}},\forall k \in [1,K]$, respectively,
\begin{equation}
\begin{aligned}
&{\cal M}_{{\bf p}_{tx,m_{tx}}}=\left\{{\bf p}_{tx,m_{tx}} \in \mathbb{R}^{2}| \,\|{\bf p}_{tx,m_{tx}} \|_2=1 \right\},\\
&\mathcal{M}_{{\bf p}_{rx,m_{rx}}}=\left\{{\bf p}_{rx,m_{rx}} \in \mathbb{R}^{2}| \,\|{\bf p}_{rx,m_{rx}} \|_2=1 \right\},\\
&\mathcal{M}_{{\bf p}_k}=\left\{{\bf p}_k \in \mathbb{R}^{2}| \,\|{\bf p}_{k} \|_2=1  \right\}.\label{manifoldP}
\end{aligned}
\end{equation}
The variables $\bf F$, $a$, and $b$ are unconstrained variables that lie on flat manifolds ${\cal M}_{{\bf F}}$, ${\cal M}_{a}$, and ${\cal M}_{b}$, respectively,
\begin{equation}
\mathcal{M}_{{\bf F}}=\left\{{\bf F} \in \mathbb{C}^{M_{rx} \times T} \right\}, \;\mathcal{M}_{a}=\left\{a \in \mathbb{R} \right\}, \; \mathcal{M}_{b}=\left\{b \in \mathbb{R} \right\}. \label{manifoldab}
\end{equation}
As discussed in \cite{boumal2023introduction}, multiple Riemannian manifolds can be combined into a product Riemannian manifold (PRM) that inherits the geometric properties of all component manifolds. Accordingly, ${\cal M}_{\bf W}$, ${\cal M}_{{\bf p}_{tx,m_{tx}}},\forall m_{tx} \in [1,M_{tx}]$, $\mathcal{M}_{{\bf p}_{rx,m_{rx}}},\forall m_{rx} \in [1,M_{rx}]$, $\mathcal{M}_{{\bf p}_k},\forall k \in [1,K]$, $\mathcal{M}_{{\bf F}}$, $\mathcal{M}_{a}$, and $\mathcal{M}_{b}$ together form the PRM \(\mathcal{M}\), defined as,
\begin{equation}
\mathcal{M} = \left\{\begin{array}{l} \mathcal{M}_{\mathbf{W}} \times \prod_{m_{tx}=1}^{M_{tx}} {\cal M}_{{\bf p}_{tx,m_{tx}}}  \times \prod_{m_{rx}=1}^{M_{rx}} {\cal M}_{{\bf p}_{rx,m_{rx}}}\\ \quad \times \prod_{k=1}^{K} {\cal M}_{{\bf p}_k}\times \mathcal{M}_{{\bf F}} \times\mathcal{M}_{a} \times \mathcal{M}_{b}\end{array}\right\}. \label{product_mani}
\end{equation}
Based on the PRM structure defined in (\ref{product_mani}), we define the optimization variable ${\bm \Psi}$ as a point in the product space $\mathcal{M}$, which is represented by an ordered tuple of elements from each individual Riemannian manifold. Consequently, we collect all decision variables into ${\bm \Psi}$ as follows,
\begin{equation}
{\bm \Psi}
=  ( {\bf W}, {\bf P}_{tx}, {\bf P}_{rx}, {\bf p}_1, \ldots, {\bf p}_K, {\bf F}, a, b ),
\end{equation}
and rewrite problem (\ref{reforfirsm2}) as the following unconstrained optimization over the PRM,
\begin{equation}
  \min_{{\bm \Psi}\in\mathcal{M}} \ \phi \triangleq\left\{\begin{array}{l} (\rho-1)a-\rho b  \\+\lambda \mu \textstyle\sum\limits_{k=1}^{K}   
 \log(1+e^{(a-\text{SINR}_k)/\mu}) \\+ \lambda \mu \textstyle\sum\limits_{t=1}^{T}   \log(1+e^{(b-\text{SCNR}_t)/\mu})
  \\+ \lambda \mu \log(1+e^{-a/\mu})+\lambda \mu \log(1+e^{-b/\mu}) \end{array}\right\} .
  \label{reforoverpro}
\end{equation}
Generally, optimizing directly on the curved PMS structure of $\mathcal{M}$ in (\ref{reforoverpro}) is non-trivial due to its inherent nonlinearity. To address this, we derive the unified tangent space (UTS), which provides a localized linear vector space approximation at any given point on the manifold. This linearization allows us to adapt standard gradient-based algorithms, typically designed for flat Euclidean spaces \cite{boumal2023introduction}, to operate effectively on the Riemannian manifold.

Firstly, consider the complex spherical manifold $\mathcal{M}_{\mathbf{W}}$ in (\ref{manifoldW}). Its tangent space at a point $\mathbf{W} \in \mathbb{C}^{M_{tx} \times (K+L_r)}$ consists of all directions orthogonal to $\mathbf{W}$ with respect to the Frobenius inner product. Mathematically, this is expressed as,
\begin{equation}
\mathrm{T}_{\mathbf{W}} \mathcal{M}_{\mathbf{W}} = \left\{{\bm \xi}_{\mathbf{W}} \in \mathbb{C}^{M_{tx} \times (K+L_r)}: \Re\left\{\operatorname{Tr}\left(\mathbf{W}^H {\bm \xi}_{\mathbf{W}} \right)\right\} = 0 \right\},
\label{tsw}
\end{equation}
where ${\bm \xi}_{\mathbf{W}}$ denotes the tangent vector. Geometrically, (\ref{tsw}) describes a hyperplane tangent to $\mathcal{M}_{\mathbf{W}}$ at $\mathbf{W}$, ensuring that any motion along ${\bm \xi}_{\mathbf{W}}$ does not violate the power constraint $\|\mathbf{W}\|_F = \sqrt{P}$, thereby staying within the manifold.

For the spherical manifolds $ \mathcal{M}_{\mathbf{p}}$ in (\ref{manifoldP}), where $\mathbf{p} \in\{\mathbf{p}_{tx,m_{tx}}, \mathbf{p}_{rx,m_{rx}},\mathbf{p}_{k} \} $, the tangent space at any point $\mathbf{p}$ on this manifold consists of all vectors orthogonal to $\mathbf{p}$. This is mathematically expressed as,
\begin{equation}
\mathrm{T}_{\mathbf{p}} \mathcal{M}_{\mathbf{p}} = \left\{ {\bm \xi}_{\mathbf{p}} \in \mathbb{R}^{2} : \mathbf{p}^T {\bm \xi}_{\mathbf{p}} = 0 \right\}.
\label{tsp}
\end{equation}
Here, ${\bm \xi}_{\mathbf{p}}$ represents the tangent vector at $\mathbf{p}$. The orthogonality condition in (\ref{tsp}) ensures that any motion along ${\bm \xi}_{\mathbf{p}}$ remains feasible by preserving the unit spherical constraint on the polarization vectors.

In the case of the flat manifolds $\mathcal{M}_{\mathbf{F}}$, $\mathcal{M}_{a}$, and $\mathcal{M}_{b}$, the variables are unconstrained Euclidean spaces. Specifically, the receive filter matrix $\mathbf{F}$ resides in $\mathbb{C}^{M_{rx} \times T}$, while the auxiliary variables $a$ and $b$ reside in $\mathbb{R}$. Their tangent spaces are isomorphic to the entire vector spaces themselves, as there is no curvature to restrict the direction of updates. These tangent spaces are given by,
\begin{equation}
\begin{aligned}
&\mathrm{T}_{\mathbf{F}} \mathcal{M}_{\mathbf{F}} = \{ \boldsymbol{\xi}_{\mathbf{F}} \in \mathbb{C}^{M_{rx} \times T} \}, \\&\mathrm{T}_{a} \mathcal{M}_{a} = \{ \xi_a \in \mathbb{R} \}, \; \mathrm{T}_{b} \mathcal{M}_{b} = \{ \xi_b \in \mathbb{R} \}.
\label{tsab}
\end{aligned}
\end{equation}

By aggregating (\ref{tsw})-(\ref{tsab}), the UTS for the full variable set $\mathbf{\Psi}$ is formed as the Cartesian product of the individual tangent spaces, given as,
\begin{equation}
\mathrm{T}_{\mathbf{\Psi}} \mathcal{M}  = \left\{\begin{array}{l} \mathrm{T}_{\mathbf{W}} \mathcal{M}_{\mathbf{W}} \times \prod_{m_{tx}=1}^{M_{tx}} \mathrm{T}_{\mathbf{p}_{tx,m_{tx}}} \mathcal{M}_{\mathbf{p}_{tx,m_{tx}}} 
 \\\quad\quad\times \prod_{m_{rx}=1}^{M_{rx}} \mathrm{T}_{\mathbf{p}_{rx,m_{rx}}} \mathcal{M}_{\mathbf{p}_{rx,m_{rx}}} \\\quad\quad\times \prod_{k=1}^{K} \mathrm{T}_{\mathbf{p}_{k}} \mathcal{M}_{\mathbf{p}_{k}} 
 \\\quad \times \mathrm{T}_{\mathbf{F}} \mathcal{M}_{\mathbf{F}}  \times \mathrm{T}_{a} \mathcal{M}_{a} \times \mathrm{T}_{b} \mathcal{M}_{b}\end{array}\right\} .
\label{UTS_combined}
\end{equation}

Leveraging the geometric structure defined in (\ref{UTS_combined}), we propose an efficient product Riemannian manifold gradient descent (PRMGD) algorithm to solve (\ref{reforoverpro}), as illustrated in Fig. \ref{manistur}. The algorithm primarily performs gradient descent steps within the UTS, followed by mapping the updated solution back onto the manifold. A key feature of this approach is the utilization of retraction operations, which project the tangent update onto the manifold surface. This methodology ensures computational efficiency while rigorously addressing the geometric constraints and curvature inherent in (\ref{reforoverpro}).

\begin{figure}[!htbp]
  \begin{center}
  \includegraphics[width=2.5in]{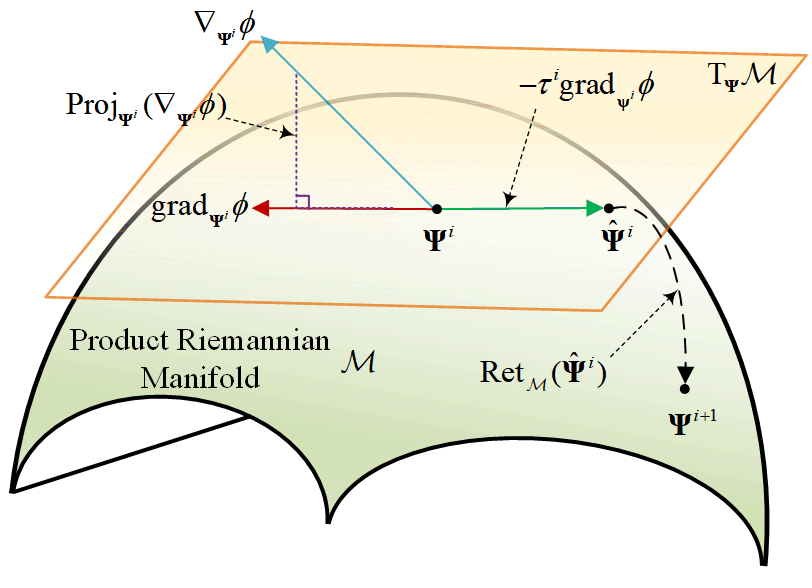}\\
  \caption{The polarimetric ISAC system with PR antennas.}\label{manistur}
  \end{center}
\end{figure}

\subsubsection{The Proposed PRMGD Algorithm}
In this subsection, we outline the steps of the PRMGD algorithm to solve the unconstrained problem defined over the product manifold $\mathcal{M}$. Let $\mathbf{\Psi}^0$ denote the initial set of optimization variables. As illustrated in Fig. \ref{manistur}, the iterative process at the \( i \)-th iteration consists of three key stages: (1) computing the Riemannian gradient in the tangent space, (2) updating the solution using retraction, and (3) adaptively adjusting the step size using the Armijo line search \cite{oviedo2022global}.

\textit{(a) Computation of the Riemannian Gradient:} 
The Riemannian gradient is obtained by orthogonally projecting the Euclidean gradient of the objective function $\phi$ onto the UTS $\mathrm{T}_{\mathbf{\Psi}^i} \mathcal{M}$. Specifically, let $\nabla_{\mathbf{\Psi}^i} \phi$ denote the Euclidean gradient evaluated at the current iterate $\mathbf{\Psi}^i \in \mathcal{M}$. The Riemannian gradient $\operatorname{grad}_{\mathbf{\Psi}^i} \phi$ is derived as,
\begin{equation}
    \operatorname{grad}_{\mathbf{\Psi}^i} \phi = \operatorname{Proj}_{\mathbf{\Psi}^i} ( \nabla_{\mathbf{\Psi}^i} \phi ),\label{Riapro31}
\end{equation}
where $\operatorname{Proj}_{\mathbf{\Psi}^i}(\cdot)$ denotes the orthogonal projection operator onto $\mathrm{T}_{\mathbf{\Psi}^i} \mathcal{M}$. Due to the product structure of the manifold, this projection decomposes into component-wise projections onto the individual tangent spaces. Consequently, the explicit expression is given by,
\begin{equation}
\begin{aligned}
    & \operatorname{Proj}_{\mathbf{\Psi}^i} ( \nabla_{\mathbf{\Psi}^i} \phi ) \\&= \left(
    \begin{array}{l}
    \operatorname{Proj}_{\mathbf{W}^i} (\nabla_{\mathbf{W}^i} \phi), \\
    \operatorname{Proj}_{\mathbf{p}^i_{tx,1}} (\nabla_{\mathbf{p}^i_{tx,1}} \phi), \dots, \operatorname{Proj}_{\mathbf{p}^i_{tx,M_{tx}}} (\nabla_{\mathbf{p}^i_{tx,M_{tx}}} \phi), \\
    \operatorname{Proj}_{\mathbf{p}^i_{rx,1}} (\nabla_{\mathbf{p}^i_{rx,1}} \phi), \dots, \operatorname{Proj}_{\mathbf{p}^i_{rx,M_{rx}}} (\nabla_{\mathbf{p}^i_{rx,M_{rx}}} \phi), \\
    \operatorname{Proj}_{\mathbf{p}^i_{1}} (\nabla_{\mathbf{p}^i_{1}} \phi), \dots, \operatorname{Proj}_{\mathbf{p}^i_{K}}(\nabla_{\mathbf{p}^i_{K}} \phi), \\
    \operatorname{Proj}_{\mathbf{F}^i} (\nabla_{\mathbf{F}^i} \phi), \operatorname{Proj}_{a^i} (\nabla_{a^i} \phi), \operatorname{Proj}_{b^i} (\nabla_{b^i} \phi)
    \end{array}
    \right),
\end{aligned}    
\end{equation}
where $\operatorname{Proj}_{\mathbf{W}^i}(\cdot)$, $\operatorname{Proj}_{\mathbf{F}^i}(\cdot)$, $\operatorname{Proj}_{a^i}(\cdot)$, and $\operatorname{Proj}_{b^i}(\cdot)$ denote the orthogonal projection operators mapping the Euclidean gradients onto the tangent spaces $\mathrm{T}_{\mathbf{W}^i} \mathcal{M}_{\mathbf{W}}$, $\mathrm{T}_{\mathbf{F}^i} \mathcal{M}_{\mathbf{F}}$, $\mathrm{T}_{a^i} \mathcal{M}_{a}$, and $\mathrm{T}_{b^i} \mathcal{M}_{b}$, respectively. Similarly, $\operatorname{Proj}_{\mathbf{p}^i}(\cdot)$ represents the projection onto $\mathrm{T}_{\mathbf{p}^i} \mathcal{M}_{\mathbf{p}}$ for any polarization vector $\mathbf{p} \in \{\mathbf{p}_{tx,m_{tx}}, \mathbf{p}_{rx,m_{rx}}, \mathbf{p}_{k}\}$. These explicit formulations are defined as,
\begin{equation}
\begin{aligned}
& \operatorname{Proj}_{\mathbf{W}^i}\left(\nabla_{\mathbf{W}^i} \phi\right) = \nabla_{\mathbf{W}^i} \phi - \Re\left\{ \operatorname{Tr}\left( \nabla_{\mathbf{W}^i} \phi (\mathbf{W}^i)^H \right) \right\} {\mathbf{W}^i}/{P},  \\
&\operatorname{Proj}_{\mathbf{p}^i}\left(\nabla_{\mathbf{p}^i} \phi\right) = \nabla_{\mathbf{p}^i} \phi - \left( (\mathbf{p}^i)^T \nabla_{\mathbf{p}^i} \phi \right) \mathbf{p}^i,  \\
&\operatorname{Proj}_{\mathbf{F}^i}(\nabla_{\mathbf{F}^i} \phi) = \nabla_{\mathbf{F}^i} \phi,  \\
&\operatorname{Proj}_{a^i}(\nabla_{a^i} \phi) = \nabla_{a^i} \phi,  \\
&\operatorname{Proj}_{b^i}(\nabla_{b^i} \phi) = \nabla_{b^i} \phi,  
\end{aligned} 
\label{Riapro33}
\end{equation}
where $\nabla_{\mathbf{W}^i} \phi$, $\nabla_{\mathbf{p}^i} \phi$, $\nabla_{\mathbf{F}^i} \phi$, $\nabla_{{a}^i} \phi$, and $\nabla_{{b}^i} \phi$ denote the respective Euclidean gradients. Specifically, the projection $\operatorname{Proj}_{\mathbf{W}^i}(\cdot)$ removes the component of the Euclidean gradient parallel to $\mathbf{W}^i$, thereby ensuring that the total power constraint is preserved in the tangent direction. Similarly, $\operatorname{Proj}_{\mathbf{p}^i}(\cdot)$ subtracts the radial component to maintain the constraint of the polarization vectors. Finally, for the variables $\mathbf{F}$, $a$, and $b$, the Riemannian gradients are identical to their Euclidean counterparts, as their respective manifolds are flat Euclidean spaces.

\textbf{\textit{Remark:}} The detailed derivations of the Euclidean gradients $\nabla_{\mathbf{W}} \phi$, $\nabla_{\mathbf{p}} \phi$ (where $\mathbf{p} \in \{\mathbf{p}_{tx,m_{tx}}, \mathbf{p}_{rx,m_{rx}}, \mathbf{p}_{k}\}$), $\nabla_{\mathbf{F}} \phi$, $\nabla_{a} \phi$, and $\nabla_{b} \phi$ are omitted here for brevity and are provided in Appendix \ref{Appendix:A}.

\textit{(b) Update Solution with Retraction:}
Similar to standard gradient descent in Euclidean space, the proposed algorithm uses the negative Riemannian gradient as the direction of steepest descent. However, due to the manifold's curvature, a direct linear update along the tangent vector results in a point that lies in the ambient Euclidean space but generally deviates from the manifold constraints \cite{boumal2023introduction}. To address this, we employ a retraction operation, which maps the updated point from the ambient space back onto the manifold surface while locally preserving the descent direction. Specifically, let $\tau^i $ denote the step size at iteration $i$. The intermediate update solution is given by,
\begin{equation}
    \mathbf{\widehat \Psi}^i = \mathbf{\Psi}^i  -\tau^i \operatorname{grad}_{\mathbf{\Psi}^i} \phi,
    \label{tangent_step}
\end{equation}
where $\mathbf{\widehat \Psi}^i$ represents the point in the ambient space obtained by moving along the tangent direction. To enforce feasibility, the next iterate $\mathbf{\Psi}^{i+1}$ is generated by applying the retraction operator $\operatorname{Ret}_{\mathcal{M}}(\cdot)$, which projects $\mathbf{\widehat \Psi}^i$ back onto the Riemannian manifold $\mathcal{M}$,
\begin{equation}
    \mathbf{\Psi}^{i+1} = \operatorname{Ret}_{\mathcal{M}} ( \mathbf{\widehat \Psi}^i ).
\end{equation}
Due to the product structure of $\mathcal{M}$, the retraction operator acts component-wise on each variable. Thus, the global retraction is expressed as,
\begin{equation}
\begin{aligned}
  &   \operatorname{Ret}_{\mathcal{M}} ( \mathbf{\widehat \Psi}^i ) \\&= \left(
    \begin{array}{l}
    \operatorname{Ret}_{\mathcal{M}_{\mathbf{W}}} (\mathbf{\widehat W}^i), \\
    \operatorname{Ret}_{\mathcal{M}_{\mathbf{p}_{tx,1}}}(\mathbf{\widehat p}_{tx,1}^i), \dots, \operatorname{Ret}_{\mathcal{M}_{\mathbf{p}_{tx,M_{tx}}}}(\mathbf{\widehat p}_{tx,M_{tx}}^i), \\
    \operatorname{Ret}_{\mathcal{M}_{\mathbf{p}_{rx,1}}}(\mathbf{\widehat p}_{rx,1}^i), \dots, \operatorname{Ret}_{\mathcal{M}_{\mathbf{p}_{rx,M_{rx}}}}(\mathbf{\widehat p}_{rx,M_{rx}}^i), \\
    \operatorname{Ret}_{\mathcal{M}_{\mathbf{p}_{1}}}(\mathbf{\widehat p}_1^i), \dots, \operatorname{Ret}_{\mathcal{M}_{\mathbf{p}_{K}}}(\mathbf{\widehat p}_K^i), \\
    \operatorname{Ret}_{\mathcal{M}_{\mathbf{F}}}(\mathbf{\widehat F}^i), \operatorname{Ret}_{\mathcal{M}_a}({\widehat a}^i), \operatorname{Ret}_{\mathcal{M}_b}({\widehat b}^i)
    \end{array}
    \right).
\end{aligned}     
\end{equation}
The specific retraction mappings for the beamforming matrix $\mathbf{W}$, the generic polarization vector $\mathbf{p} \in \{\mathbf{p}_{tx,m_{tx}}, \mathbf{p}_{rx,m_{rx}}, \mathbf{p}_{k}\}$, and the unconstrained variables $\mathbf{F}, a, b$ are defined as follows:
\begin{subequations}
\begin{align}
  &  \operatorname{Ret}_{\mathcal{M}_{\mathbf{W}}}(\mathbf{\widehat W}^i ) = \sqrt{P}{\mathbf{\widehat W}^i }/{\| \mathbf{\widehat W}^i \|_F}, \label{ret_W} \\
 & \operatorname{Ret}_{\mathcal{M}_{\mathbf{p}}}(\mathbf{\widehat p}^i )= {\mathbf{\widehat p}^i }/{\| \mathbf{\widehat p}^i \|_2}, \label{ret_p}\\
  & \operatorname{Ret}_{\mathcal{M}_{\mathbf{F}}}(\mathbf{\widehat F}^i )= \mathbf{\widehat F}^i, \;
    \operatorname{Ret}_{\mathcal{M}_{ a}}({\widehat a}^i )={\widehat a}^i,\;
    \operatorname{Ret}_{\mathcal{M}_{b}}({\widehat b}^i )= {\widehat b}^i. \label{ret_b}
\end{align}
\label{Retrecres}%
\end{subequations}
Geometrically, (\ref{ret_W}) and (\ref{ret_p}) correspond to projecting the point from the tangent plane onto the spherical surface along the radial direction. Conversely, for the unconstrained variables in (\ref{ret_b}), the retraction reduces to an identity mapping, since their underlying manifolds are flat Euclidean spaces.

\textit{(c) Step Size Adaptation:}
To ensure monotonic convergence and computational efficiency, we employ the Armijo backtracking line search strategy \cite{oviedo2022global} to dynamically adjust the step size. This mechanism allows the algorithm to take larger steps when the objective landscape is favorable and smaller steps when precise adjustments are needed. Specifically, at the $i$-th iteration, given an initial step size guess $\tau_{\text{init}}^i$, we find the smallest integer $m_i \ge 0$ such that the step size $\tau^i = \tau_{\text{init}}^i \cdot \beta^{m_i}$ satisfies the sufficient decrease condition,
\begin{equation}
  \phi(\mathbf{\Psi}^{i+1}) \le \phi(\mathbf{\Psi}^i) - c \cdot \tau^i \cdot \|\operatorname{grad}_{\mathbf{\Psi}^i} \phi\|_F^2, \label{Als}
\end{equation}
where $\beta \in (0,1)$ is the contraction factor (typically $0.5$) and $c \in (0,1)$ is a small constant. To further accelerate convergence, the initial step size $\tau_{\text{init}}^{i+1}$ for the subsequent iteration is adaptively updated based on the number of backtracking steps $m_i$ performed in the current iteration. The rationale is as follows,
\begin{itemize}
    \item If $m_i = 0$ (i.e., the condition was satisfied immediately), it implies the previous step was conservative. Thus, we aggressively increase the initial step size for the next iteration ($\tau_{\text{init}}^{i+1} = 2\tau^i$).
    \item If $m_i = 1$ (i.e., one backtrack was needed), the step size is considered appropriate, so we maintain it ($\tau_{\text{init}}^{i+1} = \tau^i$).
    \item If $m_i \ge 2$ (i.e., multiple backtracks were required), the initial guess was too optimistic. Consequently, we reduce the starting step size for the next iteration ($\tau_{\text{init}}^{i+1} = 0.5 \tau^i$).
\end{itemize}
In summary, the update rule for the initial step size of the next iteration is given by,
\begin{equation}
\tau_{\text{init}}^{i+1} = \left\{ \begin{array}{ll}
2 \tau^i, & \text{if } m_i = 0, \\
\tau^i, & \text{if } m_i = 1, \\
0.5 \tau^i, & \text{if } m_i \ge 2.
\end{array} \right. \label{usestepend}
\end{equation}

Based on the preceding analysis, the proposed PRMGD algorithm for solving problem (\ref{reforoverpro}) is summarized in Algorithm \ref{alg:1}. The algorithm terminates when either the norm of the Riemannian gradient falls below a predefined tolerance $\epsilon$ (i.e., $\|\operatorname{grad}_{\mathbf{\Psi}^i} \phi\|_F \le \epsilon^j$) or the maximum number of iterations is reached.

\begin{algorithm}[!t]
    \caption{PRMGD Algorithm to the Problem (\ref{reforoverpro})}
    \label{alg:1}
    \begin{algorithmic}[1]
        \renewcommand{\algorithmicrequire}{\textbf{Input:}}
        \renewcommand{\algorithmicensure}{\textbf{Output:}}
        
        \REQUIRE Current iterate $\mathbf{\Psi}^j$ (comprising $\mathbf{W}^j, \mathbf{P}_{tx}^j, \dots, b^j$);  Parameters $\mu^j, \lambda^j, \varepsilon^j$.
        
        \ENSURE Updated iterate $\mathbf{\Psi}^{j+1}$ for the next outer loop.
        
        \STATE \textbf{Initialization:} Inner loop counter $i = 0$ and $\mathbf{\Psi}^i = \mathbf{\Psi}^j$.
        
        \REPEAT
           \STATE Compute $\nabla \phi(\mathbf{\Psi}^i)$ via Appendix \ref{Appendix:A};
\STATE Compute $\operatorname{grad}_{\mathbf{\Psi}^i} \phi$ via projection (\ref{Riapro31})-(\ref{Riapro33});

\STATE Compute step size $\tau^i$ via (\ref{Als})-(\ref{usestepend});

\STATE Update $\mathbf{\Psi}^{i+1} =\operatorname{Ret}_{\mathcal{M}} ( \mathbf{\widehat \Psi}^i )$ via (\ref{tangent_step})-(\ref{Retrecres});
            \STATE $i \leftarrow i + 1$;
        \UNTIL $\|\operatorname{grad}_{\mathbf{\Psi}^i} \phi\|_F \le \varepsilon^j$ or $i \ge I_{inner}$.
        
    \end{algorithmic}
\end{algorithm}

\subsubsection{Dynamic Adjustment of Hyperparameters}

The performance of the proposed algorithm relies on the adjustment of the smoothing parameter $\mu$, the penalty parameter $\lambda$, the gradient tolerance $\varepsilon$, and the constraint violation threshold $\varsigma$ over outer iterations.

Firstly, $\mu$ and $\lambda$ balance approximation accuracy and feasibility. To refine the objective, $\mu$ decays via $\mu^{j+1}=\max\{\mu_{\min}, \delta_{\mu} \mu^{j}\}$. Conversely, $\lambda$ is conditionally scaled as $\lambda^{j+1} = \lambda^j / \delta_{\lambda}$ (with $\delta_{\lambda} \in (0, 1)$) only if significant constraint violations persist; this strategy enforces strict feasibility while avoiding numerical instability caused by early penalization.

Secondly, $\varepsilon$ determines the termination criterion for the inner PRMGD loop, requiring $\|\operatorname{grad}\phi\|_F \le \varepsilon^j$. As the outer loop advances, higher precision is demanded from the sub-problem. Consequently, we reduce the tolerance via $\varepsilon^{j+1} = \max\{\varepsilon_{\min}, \delta_{\varepsilon} \varepsilon^j\}$, ensuring that computational resources are allocated efficiently across iterations.

Finally, $\varsigma$ limits the maximum constraint violation, ensuring,
\begin{equation}
   \max\left\{ \begin{array}{l} 0,-a^{j+1},-b^{j+1},\\\max_{k} \{ {a^{j+1}}-\text{SINR}_k \}, \\\max_{t}\{{b^{j+1}}-\text{SCNR}_t\} \end{array}\right\} \le \varsigma^j.
\end{equation}
To drive the sequence toward the feasible region, we progressively tighten this bound via $\varsigma^{j+1} = \max\{\varsigma_{\min}, \delta_{\varsigma} \varsigma^j\}$, where $\delta_{\varsigma} \in (0, 1)$ synchronizes violation reduction with the penalty increase and $\varsigma_{\min}$ represents the ultimate feasibility target.

Algorithm \ref{alg:2} outlines the overall EP-PRMGD framework for problem (\ref{reform2euq}). We employ termination thresholds $(\mu_{\min}, \varepsilon_{\min}, \varsigma_{\min})$ as lower bounds for the smoothing parameter, gradient tolerance, and constraint violation, respectively. These thresholds are critical for balancing solution precision with computational efficiency to ensure high-quality convergence.
\begin{algorithm}[!t]
	\caption{EP-PRMGD Algorithm to the Problem (\ref{reform2euq})}
	\label{alg:2}
	\begin{algorithmic}[1]
		\renewcommand{\algorithmicrequire}{\textbf{Input:}}
		\renewcommand{\algorithmicensure}{\textbf{Output:}}
		
		\REQUIRE Initial iterate $\mathbf{\Psi}^0$; 
        Parameters $\mu^0, \lambda^0, \varepsilon^0, \varsigma^0$; 
        Bounds $\mu_{\min}, \varepsilon_{\min}, \varsigma_{\min},o_{tol}$; 
        Decay factors $\delta_{\mu}, \delta_{\lambda}, \delta_{\varepsilon}, \delta_{\varsigma} \in (0, 1)$.
        \ENSURE Optimized solution $\mathbf{\Psi}^\star$.

         \STATE \textbf{Initialization:} Outer loop counter $j=0$.
        \REPEAT
           \STATE Update $\mathbf{\Psi}^{j+1} \leftarrow \text{Algorithm \ref{alg:1} with} \;(\mathbf{\Psi}^j,  \mu^j, \lambda^j,\varepsilon^j)$;
            
            \STATE Update $\mu^{j+1} = \max\{\mu_{\min}, \mu^j \cdot \delta_{\mu}\}$;
            \STATE Update $\varepsilon^{j+1} = \max\{\varepsilon_{\min}, \varepsilon^j \cdot \delta_{\varepsilon}\}$; 
            \STATE Update $\varsigma^{j+1} = \max\{\varsigma_{\min}, \varsigma^j \cdot \delta_{\varsigma}\}$;
            
            \STATE Calculate maximum violation, 
            $$V_{\max} = \max\left\{ \begin{array}{l} 0,-a^{j+1},-b^{j+1},\\\max_{k} \{ {a^{j+1}}-\text{SINR}_k \}, \\\max_{t}\{{b^{j+1}}-\text{SCNR}_t\} \end{array}\right\};$$
            
            \IF{$V_{\max} > \varsigma^{j+1}$}
                \STATE Increase penalty $\lambda^{j+1} = \lambda^j / \delta_{\lambda}$;
            \ELSE
                \STATE Keep penalty $\lambda^{j+1} = \lambda^j$;
            \ENDIF
            
            \STATE $j \leftarrow j+1$;
            
		\UNTIL $\|\mathbf{\Psi}^j - \mathbf{\Psi}^{j-1}\| \le o_{tol}$, $\mu^j = \mu_{\min}$, and $\varsigma^j = \varsigma_{\min}$.
	\end{algorithmic}
\end{algorithm}

\subsection{Analysis of computation complexity and convergence}
\subsubsection{Analysis of Computational Complexity}
The computational complexity of the proposed EP-PRMGD framework is primarily determined by the evaluation of Euclidean gradients within the inner loop of Algorithm \ref{alg:1}. We analyze the per-iteration computational cost associated with updating the optimization variables as follows.

The dominant computational burden arises from the updates of the beamforming matrix ${\bf W}$ and the receive filter matrix ${\bf F}$. Specifically, while computing $\nabla_{\bf W}\phi$ for the communication component in (\ref{grad_W_SINR}) scales with ${\cal O}(K M_{tx} (K+L_r))$, the calculation of the sensing component in (\ref{grad_W_SCNR}) incurs a higher complexity of approximately ${\cal O}((T+C) M_{tx} M_{rx} (K+L_r))$. Additionally, calculating $\nabla_{\bf F}\phi$ appears only in the radar sensing term in (\ref{EcgraF}), contributing a complexity of ${\cal O}((T+C) M_{rx}^2 (K+L_r))$.

Regarding the polarization vectors, the computation of $\nabla_{\bf p}\phi$ is performed element-wise across the antennas. Consequently, the complexity for updating ${\bf P}_{tx}$, ${\bf P}_{rx}$, and $\{{\bf p}_k\}$ in (\ref{Egpu})-(\ref{pkSCNRE}) scales as ${\cal O}((M_{tx} + M_{rx})(T+C)(K+L_r))$. In contrast, the gradients for the auxiliary variables $a$ and $b$ in (\ref{abEUC}) incur a negligible complexity of ${\cal O}(K+T)$.

Furthermore, the Riemannian manifold operations, including the orthogonal projections in (\ref{Riapro33}) and retractions in (\ref{Retrecres}), involve low-cost vector scaling and normalizations with a linear complexity of ${\cal O}(M_{tx}(K+L_r) + M_{rx}T)$. Moreover, the step size determination via the Armijo line search \cite{oviedo2022global} in (\ref{Als}) requires repeated objective function evaluations, which share the same complexity order as the gradient derivation.

Consequently, assuming $M_{tx} \approx M_{rx} \approx M$, and denoting the number of iterations for the outer and inner loops as $I_{out}$ and $I_{in}$, respectively, the overall computational complexity of the proposed framework is approximately ${\cal O}(I_{out} I_{in} (T+C) M^2 (K+L_r))$.

\subsubsection{Analysis of Convergence}
To establish the convergence, we first demonstrate that the gradient in the inner PRMGD loop (Algorithm \ref{alg:1}) vanishes in Theorem \ref{theorem1}. Subsequently, we prove that the overall EP-PRMGD algorithm (Algorithm \ref{alg:2}) converges to a KKT point in Theorem \ref{theorem2}.

\begin{theorem} \label{theorem1} \itshape
Let $\{{\bm \Psi}^i\}$ be the sequence generated by Algorithm \ref{alg:1} with a sufficient descent step size (e.g., via Armijo backtracking). As the iteration count $i \to \infty$, the Riemannian gradient of the penalized objective function vanishes, i.e.,
\begin{equation}
     \lim_{i \to \infty} \|\operatorname{grad}\phi({\bm\Psi}^i) \|_F = 0.
\end{equation}
\end{theorem}

\quad {\bf\textit{Proof}}{\bf:} See Appendix \ref{appendixB}. $\hfill\blacksquare$

\begin{theorem} \label{theorem2} \itshape
Let $\{{\bm \Psi}^j\}$ be the sequence of iterates generated by Algorithm \ref{alg:2}. Assume that the smoothing parameter sequence $\{\mu^j\}$ converges to zero as $j \to \infty$ and that the conclusion of Theorem \ref{theorem1} holds. If a feasible limit point ${\bm \Psi}^*$ of the sequence exists, then ${\bm \Psi}^*$ is a KKT point of the original problem (\ref{reform2euq}).
\end{theorem}

\quad {\bf\textit{Proof}}{\bf:} See Appendix \ref{appendixC}. $\hfill\blacksquare$

Theorem \ref{theorem2} provides the theoretical foundation for the dynamic strategies used in Algorithm \ref{alg:2}. Specifically, a warm-start strategy stabilizes the search trajectory against ill-conditioning as $\mu^j \to 0$, while an adaptive tolerance $\varepsilon \propto \mu$ prevents computationally wasteful oversolving in the early iterations. Moreover, the penalty parameter $\lambda$ is updated conditionally according to the violation thresholds $\varsigma^j$ to ensure bounded multipliers. These mechanisms effectively balance objective minimization and feasibility enforcement, thereby directly satisfying the theoretical conditions for asymptotic convergence.

\section{Numerical Results}
This section evaluates the performance of the proposed fairness-aware beamforming design for PR-assisted ISAC systems. The proposed approach is benchmarked against the following baseline schemes:
\begin{itemize}
\item \textbf{DP-FB ($2\times$):} A fairness-aware dual-polarized design in \cite{shao2025robust}. It uses twice the number of RF chains ($M_{tx,dual} = 2M_{tx}$), serving as a performance upper bound.
\item \textbf{DP-FB ($1\times$):} A variation of the DP-FB ($2\times$) scheme, restricted to the same number of RF chains as our proposed system ($M_{tx,dual} = M_{tx}$).
\item \textbf{WOP-FB:} A conventional ISAC design without considering polarization \cite{dou2024integrated}, adapted to include fairness-aware beamforming.
\item \textbf{FP-FB:} The proposed PR-assisted ISAC architecture, where the polarization vectors $\bf p$ are predefined as fixed values without optimization. 
\item \textbf{AO-FB:} The proposed PR-assisted ISAC architecture, optimized via the AO method in \cite{lee2025integrated}.
\item \textbf{PR-WOFB:} The proposed PR-assisted ISAC architecture, optimized for aggregate sum performance without fairness constraints.
\end{itemize}

Unless otherwise specified, the simulation parameters are configured as follows. We consider an ISAC system equipped with $M=M_{tx} = M_{rx} = 32$ transmit and receive antennas, respectively, both configured as uniform linear arrays (ULAs) with half-wavelength antenna spacing. The system serves $K = 8$ communication users while concurrently transmitting $L_r = 16$ radar waveforms. The numbers of radar targets and clutter sources are set to $T=8$ and $C=8$, respectively, with their directions $\theta_q$ following the uniform distribution $\mathcal{U}(-\frac{\pi}{2},\frac{\pi}{2})$. The total transmit power budget is set to $P = 30$ dBm. The noise powers at the communication users and the sensing receiver are all set to $\sigma^2 = \sigma_k^2 = \sigma_r^2 = 20$ dBm. Finally, the ISAC trade-off parameter is chosen as $\rho=0.5$.

The channel parameters are generated following the established models in \cite{lee2025integrated} and \cite{castellanos2023linear}. Specifically, the number of communication channel paths is set to $L_k=6$ for all users. The cross-polarization leakage factor is set to $\chi_{k,l}=0.1$. The polarization phase shifts $\alpha_{k,l}^{AB}$ for $A,B\in\{H,V\}$ are assumed to be independent and uniformly distributed as $\mathcal{U}(0, 2\pi)$. The complex communication and radar path gains are given by $\beta_{k,l}\sim\mathcal{CN}(0, 1)$ and $\beta_{q}\sim\mathcal{CN}(0, 1)$, respectively.

The optimization hyperparameters are initialized as $\lambda^0=0.08$, $\mu^0=1.5$, $\varepsilon^0=10^{-2}$, and $\varsigma^0=0.1$, with respective decay rates of $\delta_{\lambda}=0.75$, $\delta_{\mu}=0.5$, $\delta_{\varepsilon}=0.6$, and $\delta_{\varsigma}=0.7$. Their corresponding lower bounds are set to $\mu_{\min}=10^{-6}$, $\varepsilon_{\min}=10^{-6}$, and $\varsigma_{\min}=10^{-5}$. The algorithm terminates when the objective variation falls below $o_{tol}=10^{-6}$, and the maximum outer and inner iterations are capped at $I_{outer} = 25$ and $I_{inner} = 150$, respectively. All results are averaged over 5,00 independent realizations.

\begin{figure}[t]
  \begin{center}
  \includegraphics[trim=0 3pt 0 10pt, clip, width=2.8in]{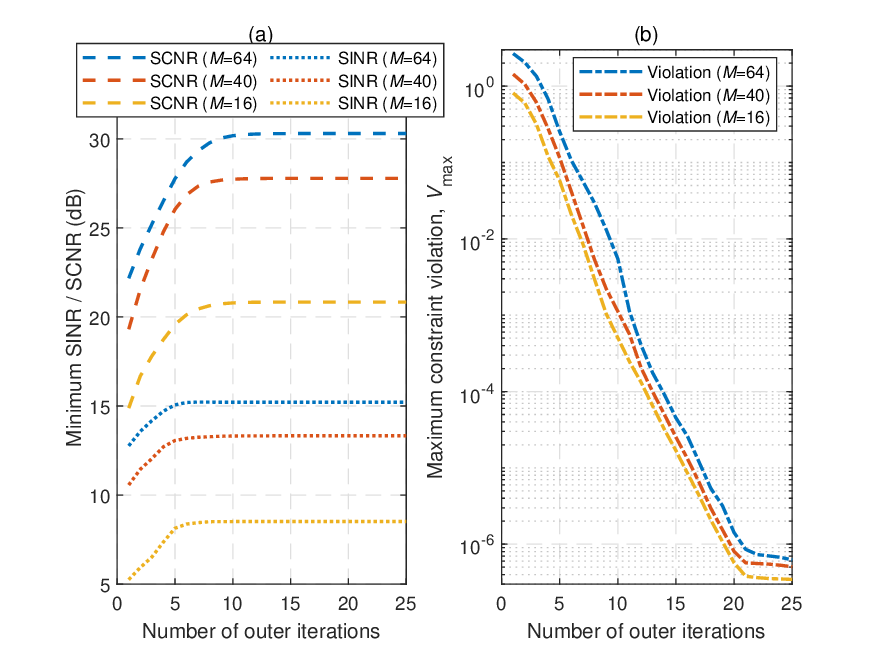}\\
  \caption{Convergence performance of the EP-PRMGD algorithm.}\label{Conv}
  \end{center}
  \vspace{-8pt}
\end{figure}

Fig. \ref{Conv} illustrates the convergence behavior of the proposed EP-PRMGD algorithm under various antenna configurations. Fig. \ref{Conv}(a) shows the worst-case SINR and SCNR versus the number of outer iterations. Both metrics increase monotonically and stabilize within 15 iterations across all antenna scales $M$, demonstrating the rapid convergence and scalability of the algorithm. Furthermore, Fig. \ref{Conv}(b) shows the maximum constraint violation $V_{\max}$ on a logarithmic scale. For all antenna configurations, $V_{\max}$ decreases sharply and falls below the tolerance threshold of $10^{-6}$ within 25 iterations, confirming the effectiveness of the EP framework in strictly enforcing the inequality constraints without numerical instability. Consequently, these empirical results corroborate the theoretical analysis in Theorem \ref{theorem2} and verify the strict convergence of the proposed EP-PRMGD algorithm.

\begin{figure}[t]
  \begin{center}
  \includegraphics[trim=0 3pt 0 15pt, clip, width=2.3in]{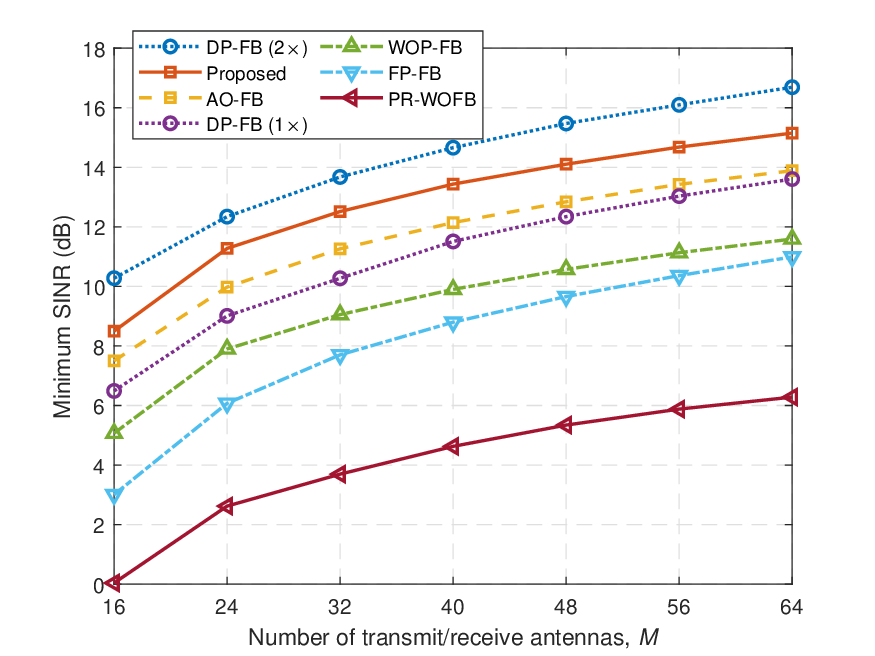}\\
  \caption{Comparison of the minimum SINR among different architectures as the number of BS transmit/receive antennas ($M$) increases from 16 to 64.}\label{DA1}
  \end{center}
    \vspace{-8pt}
\end{figure}

\begin{figure}[t]
  \begin{center}
  \includegraphics[trim=0 3pt 0 5pt, clip, width=2.3in]{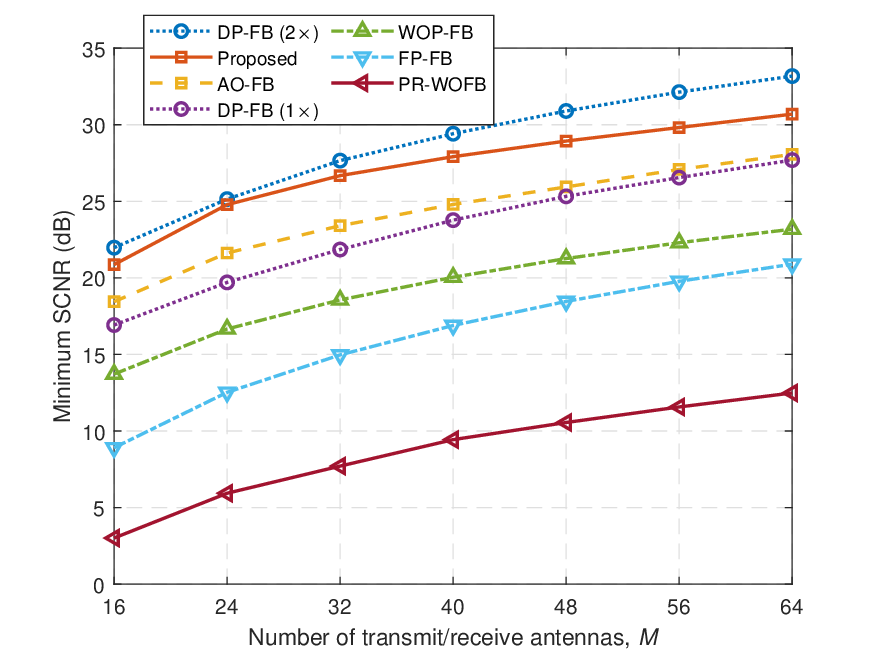}\\
  \caption{Comparison of the minimum SCNR among different architectures as the number of BS transmit/receive antennas ($M$) increases from 16 to 64.}\label{DA2}
  
  \end{center}
    \vspace{-8pt}
\end{figure}

Figs. \ref{DA1} and \ref{DA2} show the worst-case communication and sensing performance versus the number of BS antennas $M$. As expected, the minimum SINR and SCNR of all schemes increase monotonically with $M$. This gain is attributed to the enhanced spatial array gains and DoFs offered by larger arrays, which enable more flexible beam steering to suppress multi-user interference and clutter. Several observations can be made from the architectural comparison. First, the proposed PR-assisted design closely approaches the upper-bound performance of the dual-polarized DP-FB ($2\times$) system. At $M=64$, the performance gap is only 1.6 dB in SINR and 2.5 dB in SCNR, while requiring only half the RF chains. Second, under the same RF-chain cost, the proposed scheme significantly outperforms the DP-FB ($1\times$) and WOP-FB configurations, achieving gains of 1.4 dB and 3.7 dB in SINR and 3.0 dB and 7.5 dB in SCNR, respectively, at $M=64$. This confirms the structural advantage of PR antennas in providing effective DoFs for fairness-aware beamforming and improving RF-chain utilization. Third, compared with the FP-FB scheme, the proposed joint design achieves gains of more than 4.2 dB in SINR and 9.5 dB in SCNR at $M=64$, showing that the polarization DoFs offered by PR antennas can be fully exploited only through proper design. Fourth, from an algorithmic perspective, the proposed EP-PRMGD maintains a 1.4–2.5 dB advantage over AO-FB across all antenna scales, demonstrating the benefit of direct optimization over the manifold space. Finally, PR-WOFB yields the worst performance, with the minimum SINR dropping to nearly 0 dB at $M=16$, further highlighting the necessity of an explicit max-min fairness formulation to guarantee service fairness.

\begin{figure}[t]
  \begin{center}
  \includegraphics[trim=0 3pt 0 18pt, clip, width=2.3in]{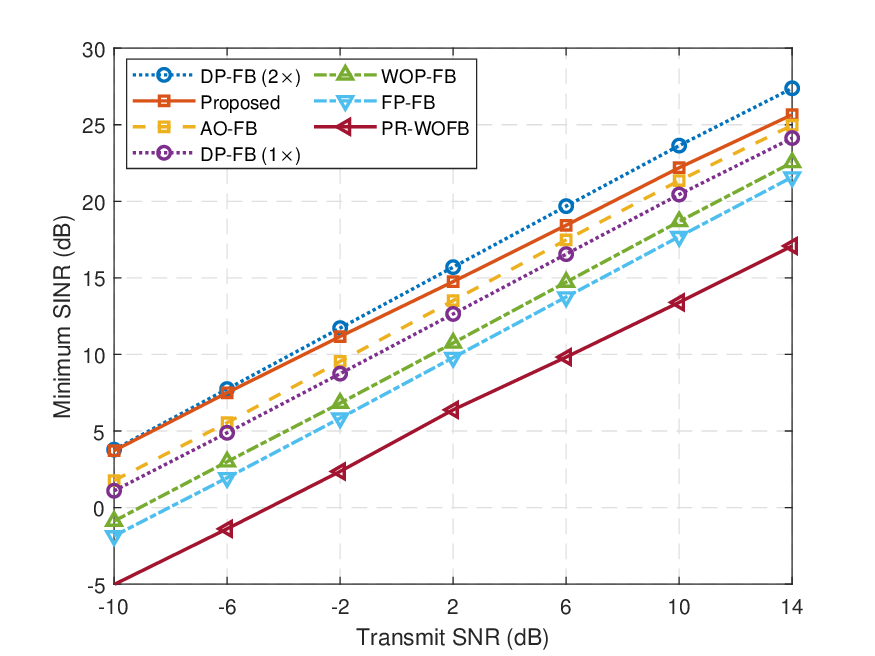}\\
  \caption{Comparison of the minimum SINR among different architectures as the transmit SNR increases from -10 to 14 dB.}\label{SNR1}
  \end{center}
    \vspace{-8pt}
\end{figure}

\begin{figure}[t]
  \begin{center}
  \includegraphics[trim=0 3pt 0 18pt, clip, width=2.3in]{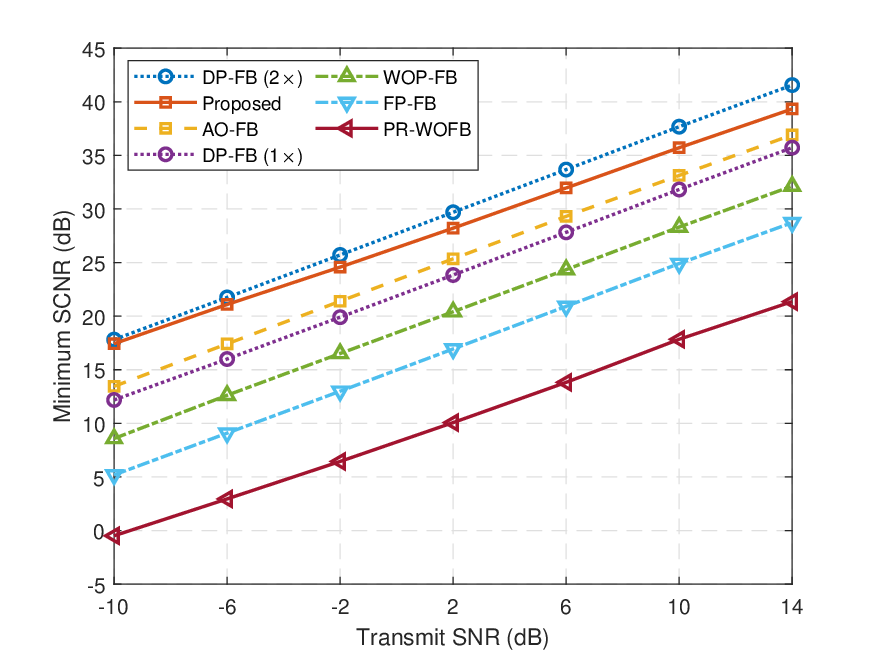}\\
  \caption{Comparison of the minimum SCNR among different architectures as the transmit SNR increases from -10 to 14 dB.}\label{SNR2}
  \end{center}
    \vspace{-8pt}
\end{figure}

Figs. \ref{SNR1} and \ref{SNR2} evaluate the worst-case performance versus the transmit SNR ($P/\sigma^2$). While all schemes exhibit the expected monotonic increase due to enhanced signal strength, the performance hierarchy remains consistent with the spatial-scaling results in Figs. \ref{DA1} and \ref{DA2}. This consistency underscores the robustness of the PR-assisted architecture across both low- and high-SNR regimes. Specifically, at a transmit SNR of 14 dB, the proposed design closely approaches the performance of the DP-FB ($2\times$) system, with marginal gaps of 1.7 dB in SINR and 2.2 dB in SCNR. Furthermore, it achieves substantial gains over the DP-FB ($1\times$) and WOP-FB configurations, reaching up to 3.1 dB in SINR and 7.2 dB in SCNR. A notable observation is the contrast with the static FP-FB scheme; the proposed joint optimization maintains a significant advantage (e.g., 11.0 dB in SCNR at 14 dB), demonstrating that dynamically optimizing polarization DoFs remains essential even with abundant transmit power. From an algorithmic perspective, EP-PRMGD sustains a 1.0–2.0 dB advantage over the AO-FB method, verifying the efficacy of the proposed algorithm across varying power budgets. Finally, the severe performance degradation of the PR-WOFB scheme, marked by an 18.0 dB SCNR loss at 14 dB, confirms that merely increasing transmit power cannot resolve spatial resource imbalances, thereby necessitating an explicit fairness formulation.

\begin{figure}[t]
  \begin{center}
  \includegraphics[trim=0 3pt 0 18pt, clip, width=2.3in]{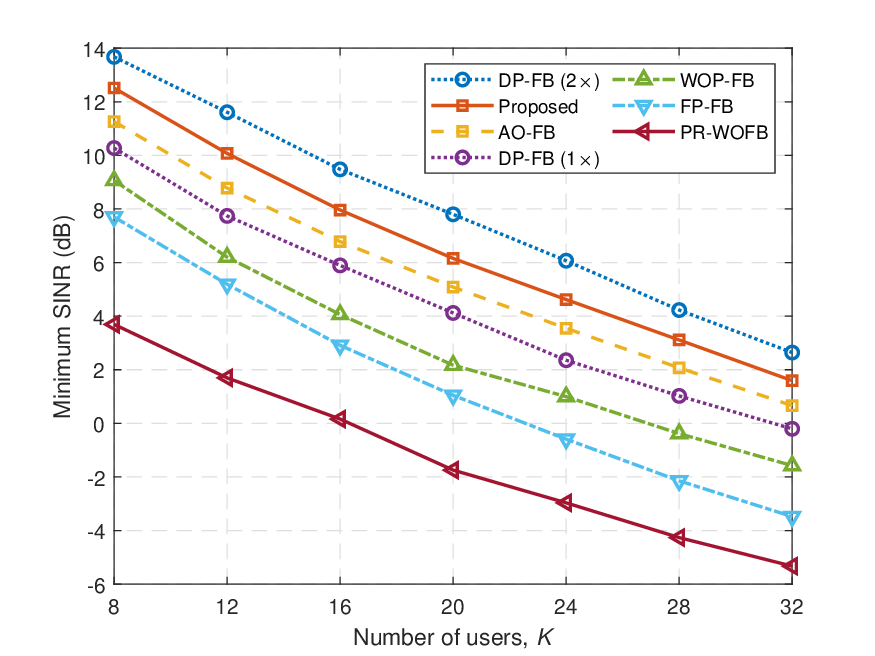}\\
  \caption{Comparison of the minimum SINR among different architectures as the number of users ($K$) increases from 8 to 32.}\label{User}
  \end{center}
  \vspace{-8pt}
\end{figure}

\begin{figure}[t]
  \begin{center}
  \includegraphics[trim=0 3pt 0 18pt, clip,width=2.3in]{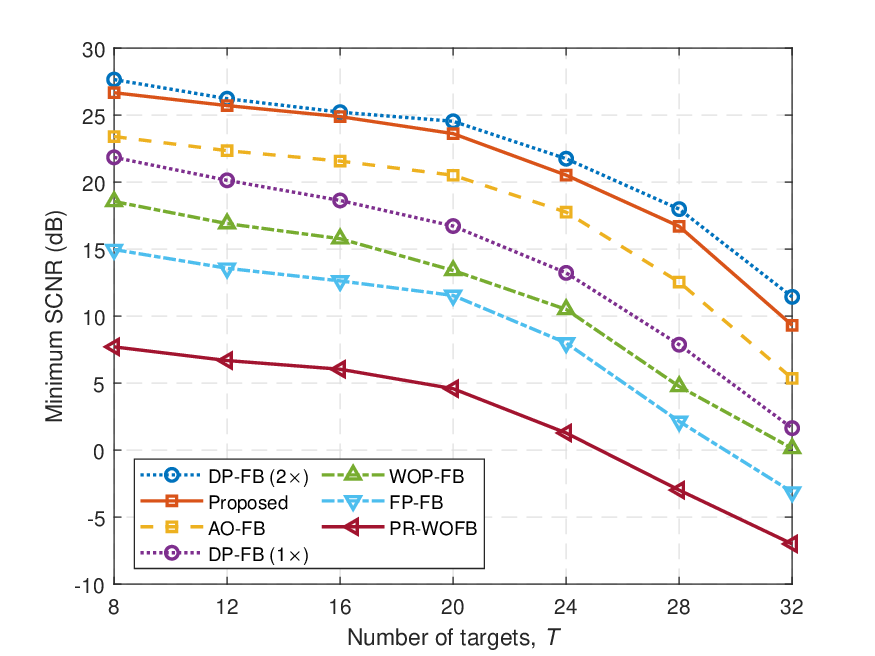}\\
  \caption{Comparison of the minimum SCNR among different architectures as the number of targets ($T$) increases from 8 to 32.}\label{Tar}
  \end{center}
  \vspace{-8pt}
\end{figure}

Figs. \ref{User} and \ref{Tar} further evaluate the worst-case performance from a system-load perspective, with results plotted versus the number of communication users $K$ and sensing targets $T$, respectively. As the system load increases, the minimum SINR and SCNR of all schemes degrade monotonically because the limited transmit power and spatial DoFs must be shared among more nodes, thereby intensifying resource competition. At $K=32$ and $T=32$, the proposed design still closely tracks the upper-bound performance of the dual-polarized DP-FB ($2\times$) system, with marginal gaps of only 1.0 dB in SINR and 1.5 dB in SCNR. It also significantly outperforms the DP-FB ($1\times$), WOP-FB, and FP-FB configurations, achieving gains of approximately 1.7, 3.1, and 5.0 dB in SINR and 7.7, 9.2, and 12.4 dB in SCNR, respectively. This demonstrates the effectiveness of dynamically exploiting polarization DoFs for alleviating resource bottlenecks in congested networks. Moreover, the proposed EP-PRMGD retains a consistent advantage over the AO-FB method, confirming its robustness under heavy loads. 

\begin{figure}[t]
  \begin{center}
  \includegraphics[trim=0 3pt 0 18pt, clip, width=2.3in]{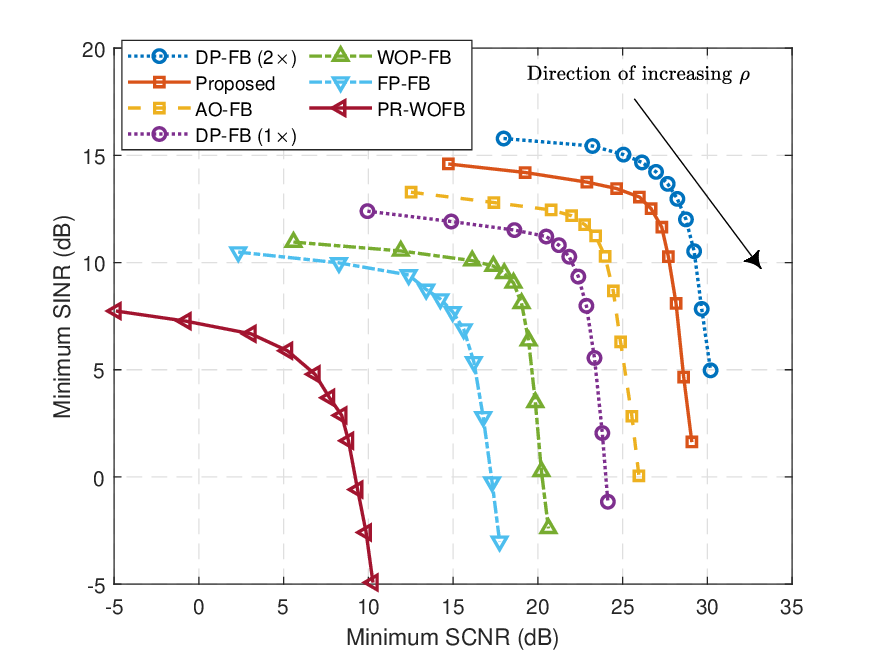}\\
  \caption{Comparison of the minimum SINR and minimum SCNR trade-off among different architectures.}\label{tradeoff}
  \end{center}
  \vspace{-8pt}
\end{figure}

Finally, Fig. \ref{tradeoff} investigates the fundamental trade-off between communication and sensing by varying the trade-off parameter $\rho$ uniformly from 0 to 1 with a step size of 0.1. As expected, all schemes exhibit a clear Pareto-like boundary. Specifically, prioritizing one functionality inevitably degrades the other because of intense competition for shared resources. Notably, the discrete operating points are densely clustered around the curved portion of the boundary. This non-uniform distribution arises because the linear variation of $\rho$ maps nonlinearly to the highly coupled max-min objective space. Physically, this region corresponds to the most resource-efficient operating regime. Outside this region, a marginal improvement in one metric requires a disproportionately large sacrifice in the other. Therefore, the trade-off factor and system configuration parameters should be carefully calibrated to strike an effective balance in practical fairness-aware ISAC deployments with PR antennas.

\section{Conclusion}
This paper investigated a fairness-aware beamforming design for polarimetric ISAC systems equipped with cost-effective PR antennas. A joint optimization framework was developed to maximize the minimum communication SINR and sensing SCNR by designing the transmit beamformer, PR control coefficients, and receive filters. To tackle the highly coupled and nonsmooth problem, an EP-PRMGD algorithm was proposed, enabling direct optimization over a curved space without artificial block-wise decoupling. The proposed algorithm was theoretically proven to converge to a KKT point. Extensive simulations validated that the PR-assisted scheme closely approaches the performance of fully dual-polarized architectures while requiring only half the RF chains, demonstrating its superiority in balancing hardware efficiency and service equity.

\appendices
\section{Computation of the Euclidean gradients}
\label{Appendix:A}
To obtain the Euclidean gradients of $\phi({\bm\Psi})$ in (\ref{reforoverpro}), we distinguish between the design variables $\mathbf{x} \in \{\mathbf{W}, \mathbf{F}, \mathbf{p}\}$ and the auxiliary variables $\{a, b\}$. By denoting $\mathcal{S}(z) \triangleq \mu \log(1 + e^{z/\mu})$, the objective function is given as,
\begin{equation}
\begin{aligned}
\phi({\bm \Psi}) &= (\rho - 1)a - \rho b + \lambda \textstyle\sum_{k=1}^{K} \mathcal{S}(a - \text{SINR}_k) \\
&+ \lambda \textstyle\sum_{t=1}^{T} \mathcal{S}(b - \text{SCNR}_t) +  \lambda \mathcal{S}(-a)+\lambda \mathcal{S}(-b).
\end{aligned} 
\label{phi_struct}
\end{equation}
For the auxiliary variables $a$ and $b$, the gradients are derived directly from the partial derivatives in (\ref{phi_struct}). 
For the system design variables $\mathbf{x}$, which influence the objective indirectly via the SINR and SCNR metrics, we apply the chain rule, 
\begin{equation}
\nabla_{\mathbf{x}} \phi = - \textstyle\sum_{k=1}^K \alpha_k \nabla_{\mathbf{x}} \text{SINR}_k - \textstyle\sum_{t=1}^T \beta_t \nabla_{\mathbf{x}} \text{SCNR}_t,
\label{chain_rule_gen}
\end{equation}
where the weighting coefficients $\alpha_k$ and $\beta_t$ are given by the derivative of the smoothing function $\sigma(z/\mu) \triangleq (1+e^{-z/\mu})^{-1}$, yielding,
\begin{equation}
\alpha_k = \lambda \cdot \sigma(\frac{a - \text{SINR}_k}{\mu}), \quad \beta_t = \lambda \cdot \sigma(\frac{b - \text{SCNR}_t}{\mu}).
\end{equation}
In the following, we provide the detailed derivations.

\textit{1) Euclidean Gradients of the Beamforming Matrix} $\mathbf{W}$:
For $\mathbf{W} \in \mathbb{C}^{M_{tx} \times (K+L_r)}$, we derive its gradient by analyzing the SINR and SCNR terms separately. For the user $k$, we define $\mathbf{h}_{eff,k}^H = \mathbf{p}_k^T \mathbf{H}_k \mathbf{P}_{tx} \in \mathbb{C}^{1 \times M_{tx}}$. The $\text{SINR}_k$ is defined as the ratio $\mathcal{N}_k / \mathcal{D}_k$, where $\mathcal{N}_k = |\mathbf{h}_{eff,k}^H \mathbf{w}_k|^2$ and $\mathcal{D}_k = \sum_{j \neq k} |\mathbf{h}_{eff,k}^H \mathbf{w}_j|^2 + \sigma_k^2$. To compute the gradient with respect to the $j$-th column $\mathbf{w}_j$, we use the quotient rule, 
\begin{equation}
    \nabla(f/g) = g^{-2}(g \nabla f - f \nabla g) = g^{-1}(\nabla f - (f/g) \nabla g),\label{ quotientrule}
\end{equation}
considering the following two cases,
\begin{itemize}
    \item \textbf{Case $j=k$:} The variable $\mathbf{w}_k$ appears only in $\mathcal{N}_k$. Thus, $\nabla_{\mathbf{w}_k} \mathcal{N}_k = 2 \mathbf{h}_{eff,k} (\mathbf{h}_{eff,k}^H \mathbf{w}_k)$ and $\nabla_{\mathbf{w}_k} \mathcal{D}_k = \mathbf{0}$.
    \item \textbf{Case $j \neq k$:} The variable $\mathbf{w}_j$ appears only in $\mathcal{D}_k$. Thus, $\nabla_{\mathbf{w}_j} \mathcal{N}_k = \mathbf{0}$ and $\nabla_{\mathbf{w}_j} \mathcal{D}_k = 2 \mathbf{h}_{eff,k} (\mathbf{h}_{eff,k}^H \mathbf{w}_j)$.
\end{itemize}
Applying the quotient rule in (\ref{ quotientrule}), we obtain the column-wise gradients,
\begin{equation}
\nabla_{\mathbf{w}_j} \text{SINR}_k = 
\begin{cases} 
\frac{2}{\mathcal{D}_k} \mathbf{h}_{eff,k} \mathbf{h}_{eff,k}^H \mathbf{w}_k, & j = k, \\
 -\frac{2 \cdot \text{SINR}_k}{\mathcal{D}_k} \mathbf{h}_{eff,k} \mathbf{h}_{eff,k}^H \mathbf{w}_j, & j \neq k.
\end{cases}
\label{grad_W_SINR}
\end{equation}

For sensing, let $\mathbf{u}_{t,q}^H = \mathbf{f}_t^H \mathbf{P}_{rx}^T \mathbf{G}_q \mathbf{P}_{tx} \in \mathbb{C}^{1 \times M_{tx}}$ denote the effective response vector of the object $q$. The SCNR is expressed as $\mathcal{N}_t^r / \mathcal{D}_t^r$, where $\mathcal{N}_t^r = \|\mathbf{u}_{t,t}^H \mathbf{W}\|^2$ and $\mathcal{D}_t^r = \sum_{q \neq t} \|\mathbf{u}_{t,q}^H \mathbf{W}\|^2 + \sigma_r^2 \|\mathbf{f}_t\|^2$.
Consequently, the gradients of the numerator and denominator are,
\begin{equation}
    \nabla_{\mathbf{W}} \mathcal{N}_t^r = 2 \mathbf{u}_{t,t} \mathbf{u}_{t,t}^H \mathbf{W}, \quad
\nabla_{\mathbf{W}} \mathcal{D}_t^r = \textstyle\sum_{q \neq t} 2 \mathbf{u}_{t,q} \mathbf{u}_{t,q}^H \mathbf{W}.
\end{equation}
Applying the quotient rule in (\ref{ quotientrule}), we arrive at,
\begin{equation}
\nabla_{\mathbf{W}} \text{SCNR}_t = \frac{2}{\mathcal{D}_t^r} ( \mathbf{u}_{t,t}\mathbf{u}_{t,t}^H \mathbf{W} - \text{SCNR}_t \sum_{q \neq t} \mathbf{u}_{t,q}\mathbf{u}_{t,q}^H \mathbf{W} ).
\label{grad_W_SCNR}
\end{equation}
Finally, substituting (\ref{grad_W_SINR}) and (\ref{grad_W_SCNR}) into the chain rule expression (\ref{chain_rule_gen}) yields the complete Euclidean gradient $\nabla_{\mathbf{W}} \phi$.

\textit{2) Euclidean Gradients of the Receive Filter} $\mathbf{F}$: 
The receive filter matrix $\mathbf{F} \in \mathbb{C}^{M_{rx} \times T}$ is decoupled across targets, meaning the $t$-th column vector $\mathbf{f}_t \in \mathbb{C}^{M_{rx} \times 1}$ only affects $\text{SCNR}_t$. We define $\mathbf{A}_t = \mathbf{S}_t \mathbf{S}_t^H \in \mathbb{C}^{M_{rx} \times M_{rx}}$ and $\mathbf{B}_t = \sum_{q \neq t} \mathbf{S}_q \mathbf{S}_q^H + \sigma_r^2 \mathbf{I}_{M_{rx}} \in \mathbb{C}^{M_{rx} \times M_{rx}}$, where $\mathbf{S}_q = \mathbf{P}_{rx}^T \mathbf{G}_q \mathbf{P}_{tx} \mathbf{W} \in \mathbb{C}^{M_{rx} \times (K+L_r)}$. The $\text{SCNR}_t$ is given as,
\begin{equation}
    \frac{\mathbf{f}_t^H \mathbf{A}_t \mathbf{f}_t}{\mathbf{f}_t^H \mathbf{B}_t \mathbf{f}_t},
\end{equation}
by applying the rule in (\ref{ quotientrule}), its gradient is given by,
\begin{equation}
\nabla_{\mathbf{f}_t} \text{SCNR}_t = \frac{2}{\mathbf{f}_t^H \mathbf{B}_t \mathbf{f}_t} \left( \mathbf{A}_t \mathbf{f}_t - \text{SCNR}_t \mathbf{B}_t \mathbf{f}_t \right).\label{EcgraF}
\end{equation}
Consequently, using (\ref{chain_rule_gen}), we obtain $\nabla_{\mathbf{f}_t} \phi = -\beta_t \nabla_{\mathbf{f}_t} \text{SCNR}_t$.

\textit{3) Euclidean Gradients of the Polarization Vectors} $\mathbf{p}$: 
Since the polarization vectors $\mathbf{p} \in \mathbb{R}^{2}$ are real-valued, we use the chain rule identity: given a scalar function $f(\mathbf{p}) = |g(\mathbf{p})|^2$, its gradient is $\nabla_{\mathbf{p}} f = 2 \Re \{ g(\mathbf{p})^* \nabla_{\mathbf{p}} g(\mathbf{p}) \}$.

\textbf{User Polarization $\mathbf{p}_k$}: 
The polarization vector $\mathbf{p}_k \in \mathbb{R}^2$ influences only the $\text{SINR}_k$ of user $k$. By defining $\mathbf{s}_{k,j} = \mathbf{H}_k \mathbf{P}_{tx} \mathbf{w}_j \in \mathbb{C}^{2 }$ and $y_{k,j} = \mathbf{p}_k^T \mathbf{s}_{k,j}$, the numerator and denominator of $\text{SINR}_k$ are rewritten as $\mathcal{N}_k=|y_{k,k}|^2$ and $\mathcal{D}_k = \sum_{j \neq k} |y_{k,j}|^2 + \sigma_k^2$, respectively. By applying a logic similar to that in (\ref{grad_W_SINR}), we arrive at,
\begin{equation}
\begin{aligned}
&\nabla_{\mathbf{p}_k} \text{SINR}_k = \frac{1}{\mathcal{D}_k} \nabla_{\mathbf{p}_k} \mathcal{N}_k - \frac{\text{SINR}_k}{\mathcal{D}_k} \nabla_{\mathbf{p}_k} \mathcal{D}_k \\
&= \frac{2}{\mathcal{D}_k} \Re\left( y_{k,k}^* \mathbf{s}_{k,k} \right) - \frac{2 \cdot \text{SINR}_k}{\mathcal{D}_k} \textstyle\sum_{j \neq k} \Re\left( y_{k,j}^* \mathbf{s}_{k,j} \right).\label{Egpu}
\end{aligned}
\end{equation}

\textbf{Transmit Polarization $\mathbf{p}_{tx, m_{tx}}$}: 
The vector $\mathbf{p}_{tx, m_{tx}} \in \mathbb{R}^2$ with $m_{tx} \in [1, M_{tx}]$ influences both communication and sensing metrics. For communication, we define $\bm{\eta}_{k,m_{tx}} = (\mathbf{p}_k^T \mathbf{H}_{k,m_{tx}})^T \in \mathbb{C}^{2 }$, where $\mathbf{H}_{k,m_{tx}} \in \mathbb{C}^{2 \times 2}$ denotes the sub-matrix of $\mathbf{H}_k$ corresponding to the $m_{tx}$-th transmit antenna. The partial derivative of $y_{k,j}$ is given by $\frac{\partial y_{k,j}}{\partial \mathbf{p}_{tx,m_{tx}}} = \bm{\eta}_{k,m_{tx}} w_{j,m_{tx}}$. Similarly to (\ref{Egpu}), the gradient is expressed as,
\begin{equation}
\begin{aligned}
\nabla_{\mathbf{p}_{tx,m_{tx}}} &\text{SINR}_k = \frac{2}{\mathcal{D}_k} \Re\left( y_{k,k}^* \bm{\eta}_{k,m_{tx}} w_{k,m_{tx}} \right) \\
&- \frac{2 \cdot \text{SINR}_k}{\mathcal{D}_k} \textstyle\sum_{j \neq k} \Re\left( y_{k,j}^* \bm{\eta}_{k,m_{tx}} w_{j,m_{tx}} \right).
\end{aligned}
\end{equation}
For sensing, let $\mathbf{G}_{q,m_{tx}} \in \mathbb{C}^{2M_{tx} \times 2}$ be the $m_{tx}$-th block column of $\mathbf{G}_q$. We define $\mathbf{d}_{t,q,m_{tx}} = (\mathbf{f}_t^H \mathbf{P}_{tx}^T \mathbf{G}_{q,m_{tx}})^T \in \mathbb{C}^{2 }$ and $z_{t,q,j} = \mathbf{u}_{t,q}^H \mathbf{w}_j$. The numerator and denominator of $\text{SCNR}_t$ are $\mathcal{N}_t^r = \sum_{j} |z_{t,t,j}|^2$ and $\mathcal{D}_t^r = \sum_{q \neq t} \sum_{j} |z_{t,q,j}|^2 + \sigma_r^2 \|\mathbf{f}_t\|^2$, respectively. The gradient is derived as,
\begin{equation}
\begin{aligned}
\nabla_{\mathbf{p}_{tx,m_{tx}}} &\text{SCNR}_t = \frac{2}{\mathcal{D}_t^r} \Re(\textstyle \sum_{j} z_{t,t,j}^* \mathbf{d}_{t,t,m} w_{j,m_{tx}}) \\
&- \frac{2 \cdot \text{SCNR}_t}{\mathcal{D}_t^r} \textstyle\sum_{q \neq t} \Re( \textstyle\sum_{j} z_{t,q,j}^* \mathbf{d}_{t,q,m} w_{j,m_{tx}} ),
\end{aligned}
\end{equation}

\textbf{Receive Polarization $\mathbf{p}_{rx, m_{rx}}$}: 
The vector $\mathbf{p}_{rx, m_{rx}} \in \mathbb{R}^2$ with $m_{rx} \in [1, M_{rx}]$ influences only the sensing metric $\text{SCNR}_t$. We define $\mathbf{V}_{q} = \mathbf{G}_q \mathbf{P}_{tx} \mathbf{W} \in \mathbb{C}^{2M_{rx}\times (K+L_r)}$ and let $\mathbf{v}_{q,m_{rx},j} \in \mathbb{C}^{2 \times 1}$ denote the $j$-th column of the $m_{rx}$-th block row of $\mathbf{V}_q$. Additionally, let $f_{t,m_{rx}}$ denote the $m_{rx}$-th element of the receive filter vector $\mathbf{f}_t$. The partial derivative of $z_{t,q,j} = \mathbf{u}_{t,q}^H \mathbf{w}_j$ with respect to $\mathbf{p}_{rx,m_{rx}}$ is given by $\frac{\partial z_{t,q,j}}{\partial \mathbf{p}_{rx,m_{rx}}} = f_{t,m_{rx}}^* \mathbf{v}_{q,m_{rx},j}$. By applying a similar logic to the previous cases, the gradient is derived as,
\begin{equation}
\begin{aligned}
 &\nabla_{\mathbf{p}_{rx,m_{rx}}}\text{SCNR}_t = \frac{2}{\mathcal{D}_t^r} \Re( \textstyle\sum_{j} z_{t,t,j}^* f_{t,m_{rx}}^* \mathbf{v}_{t,m_{rx},j} ) \\
&- \frac{2 \cdot \text{SCNR}_t}{\mathcal{D}_t^r} \textstyle\sum_{q \neq t} \Re( \textstyle\sum_{j} z_{t,q,j}^* f_{t,m_{rx}}^* \mathbf{v}_{q,m_{rx},j} ).\label{pkSCNRE}
\end{aligned}
\end{equation}

\textit{4) Euclidean Gradients of the Auxiliary Variables} $a, b$: 
The gradients for $a$ and $b$ are directly obtained from (\ref{phi_struct}) as,
\begin{subequations}
\begin{align}
\nabla_{a} \phi &= (\rho - 1) + \textstyle\sum_{k=1}^K \alpha_k - \lambda \sigma(-a/\mu), \\
\nabla_{b} \phi &= -\rho + \textstyle\sum_{t=1}^T \beta_t - \lambda \sigma(-b/\mu).
\end{align}
\label{abEUC}%
\end{subequations}
This completes the derivation of all Euclidean gradients.  $\hfill\blacksquare$

\section{Proof of Theorem \ref{theorem1}}
\label{appendixB}
The proof relies on a standard telescoping sum argument utilizing the sufficient decrease condition provided by the Armijo line search in (\ref{Als}) . First, assuming the step sizes produced by the line search are bounded away from zero, i.e., $\tau^i \ge \tau_{\min} > 0$, the Armijo condition in (\ref{Als}) implies,
\begin{equation}
  \phi(\mathbf{\Psi}^i)-\phi(\mathbf{\Psi}^{i+1}) \ge   c \cdot \tau_{\min} \cdot \|\operatorname{grad}\phi(\mathbf{\Psi}^i)\|_F^2.
\end{equation}
Summing this inequality over the iterations from $i=0$ to $I_{\max}-1$, we obtain,
\begin{equation}
    \begin{aligned}
        \phi(\mathbf{\Psi}^0) - \phi(\mathbf{\Psi}^{I_{\max}}) &= \sum\limits_{i=0}^{I_{\max}-1} (\phi(\mathbf{\Psi}^i) - \phi(\mathbf{\Psi}^{i+1})) \\
        &\ge c \cdot \tau_{\min} \sum\limits_{i=0}^{I_{\max}-1} \|\operatorname{grad}\phi(\mathbf{\Psi}^{i}) \|_F^2.
    \end{aligned}
\end{equation}
Since the objective function $\phi$ is bounded below by some value $\phi^{\text{low}}$, we have $\phi(\mathbf{\Psi}^0) - \phi(\mathbf{\Psi}^{I_{\max}}) \le \phi(\mathbf{\Psi}^0) - \phi^{\text{low}}$. Taking the limit as $I_{\max} \to \infty$, the series on the right-hand side converges,
\begin{equation}
    \sum\limits_{i=0}^{\infty} \|\operatorname{grad}\phi(\mathbf{\Psi}^{i}) \|_F^2 \le \frac{\phi(\mathbf{\Psi}^0) - \phi^{\text{low}}}{c \cdot \tau_{\min}} < \infty.
\end{equation}
The convergence of the series implies that the summands must vanish asymptotically, i.e.,
\begin{equation}
    \lim_{i \to \infty} \|\operatorname{grad}\phi(\mathbf{\Psi}^{i}) \|_F = 0.
\end{equation}
This completes the proof. \hfill $\blacksquare$

\section{Proof of Theorem \ref{theorem2}}
\label{appendixC}
In the $j$-th iteration of Algorithm \ref{alg:2}, the penalized objective (\ref{reform2euq}) is defined as,
\begin{equation}
     \phi(\bm{\Psi}) = f_0(\bm{\Psi}) + \lambda^j \sum_{m \in \mathcal{C}} \mu^j \log( 1 + \exp(\frac{g_m(\bm{\Psi})}{\mu^j}) ),
\end{equation}
where $f_0(\bm{\Psi})$ is the primary objective, $g_m(\bm{\Psi}) \le 0$ denotes the $m$-th constraint, and $\mathcal{C}$ is the constraint index set.

Let $\{\bm{\Psi}^j\}$ denote the iterate sequence generated by Algorithm~\ref{alg:2}. The compactness of the search space $\mathcal{M}$ implies $\{\bm{\Psi}^j\}$ is bounded. By the Bolzano--Weierstrass theorem \cite{oman2017short}, there exists a subsequence $\{\bm{\Psi}^{j_k}\}$ converging to a limit point $\bm{\Psi}^*$. Assuming feasibility of $\bm{\Psi}^*$ as per the theorem's premise, we verify the KKT conditions term by term.

\subsection*{1. Multiplier Convergence}
First, we establish penalty parameter stabilization and multiplier convergence. According to Algorithm \ref{alg:2}, $\lambda$ updates as,
\begin{equation}
    \lambda^{j+1} = 
    \begin{cases} 
    \lambda^j/ \delta_{\lambda}, & \text{if } V_{\max}(\bm{\Psi}^{j+1}) > \varsigma^{j+1}, \\
    \lambda^j, & \text{otherwise},
    \end{cases}
\end{equation}
where $\delta_{\lambda}\in(0, 1)$.
Since $\bm{\Psi}^*$ is feasible, the constraint violation along the convergent subsequence vanishes,
\begin{equation}
    \lim_{k \to \infty} V_{\max}(\bm{\Psi}^{j_k}) = 0.
\end{equation}
Assuming the tolerance sequence $\{\varsigma^j\}$ decays sufficiently slowly relative to the constraint violation rate, there exists a finite index $J_0$ such that,
\begin{equation}
    V_{\max}(\bm{\Psi}^{j_k+1}) \le \varsigma^{j_k+1}, \quad \forall j_k \ge J_0.
\end{equation}
Consequently, the penalty parameter stabilizes at a finite constant,
\begin{equation} \label{eq:lambda_stabilize}
    \exists \lambda^* < \infty \text{ s.t. } \lambda^{j_k} = \lambda^*, \quad \forall j_k \ge J_0.
\end{equation}

Next, consider the implicit multipliers derived from the smoothing function gradient,
\begin{equation}
    \gamma_m^j \triangleq \frac{\partial \mathcal{S}(g_m(\bm{\Psi}^j))}{\partial g_m}
    = \frac{\exp(g_m(\bm{\Psi}^j)/\mu^j)}{1 + \exp(g_m(\bm{\Psi}^j)/\mu^j)}.\label{use64}
\end{equation}
Examining the Riemannian gradient of $\phi$ reveals the effective multiplier structure,
\begin{equation}
    \operatorname{grad} \phi(\bm{\Psi}^j) = \operatorname{grad} f_0(\bm{\Psi}^j) + \sum_{m \in \mathcal{C}} \underbrace{\lambda^j \gamma_m^j}_{\text{eff. mult.}} \operatorname{grad} g_m(\bm{\Psi}^j).
\end{equation}
The function in (\ref{use64}) strictly bounds the vector $\bm{\gamma}^j$ within the unit hypercube,
\begin{equation}
    \bm{\gamma}^j \in (0, 1)^{|\mathcal{C}|} \subset [0, 1]^{|\mathcal{C}|}, \quad \forall j.
\end{equation}
Since the closure $[0, 1]^{|\mathcal{C}|}$ is compact, applying Bolzano--Weierstrass yields a convergent subsequence,
\begin{equation} \label{eq:gamma_limit}
    \lim_{k \to \infty} \bm{\gamma}^{j_k} = \bm{\gamma}^* \in [0, 1]^{|\mathcal{C}|}.
\end{equation}
Combining \eqref{eq:lambda_stabilize} and \eqref{eq:gamma_limit}, the effective multipliers converge,
\begin{equation} \label{eq:nu_converge}
    \nu_m^{j_k} \triangleq \lambda^{j_k} \gamma_m^{j_k} \xrightarrow{k \to \infty} \nu_m^* \triangleq \lambda^* \gamma_m^*, \quad \forall m \in \mathcal{C},
\end{equation}
which ensures dual feasibility,
\begin{equation} \label{eq:dual_feas}
    \lambda^* > 0 \land \gamma_m^* \ge 0 \implies \nu_m^* \ge 0.
\end{equation}

\subsection*{2. Complementary Slackness}
We verify $\nu_m^* g_m(\bm{\Psi}^*) = 0$ by partitioning $\mathcal{C}$ into active ($\mathcal{A}$) and inactive ($\mathcal{I}$) sets at $\bm{\Psi}^*$,
\begin{equation}
    \mathcal{A} \triangleq \{m \in \mathcal{C} \mid g_m(\bm{\Psi}^*) = 0\}, \quad 
    \mathcal{I} \triangleq \{m \in \mathcal{C} \mid g_m(\bm{\Psi}^*) < 0\}.
\end{equation}

\subsubsection*{Case 1: Active Constraints ($m \in \mathcal{A}$)}
For $m \in \mathcal{A}$, the condition holds trivially,
\begin{equation}
    \nu_m^* g_m(\bm{\Psi}^*) = \nu_m^* \cdot 0 = 0.
\end{equation}

\subsubsection*{Case 2: Inactive Constraints ($m \in \mathcal{I}$)}
For $m \in \mathcal{I}$, strict feasibility holds: $g_m(\bm{\Psi}^*) < 0$. 
Since each $g_m(\cdot)$ is continuous and $\bm{\Psi}^{j_k} \to \bm{\Psi}^*$, there exists $\delta > 0$ and $K_0$ such that for all $k>K_0$,
\begin{equation}
    g_m(\bm{\Psi}^{j_k}) \le -\delta < 0.
\end{equation}
As $\mu^{j_k} \to 0$, the argument of the exponential term diverges to negative infinity,
\begin{equation}
    \frac{g_m(\bm{\Psi}^{j_k})}{\mu^{j_k}} \le \frac{-\delta}{\mu^{j_k}} \xrightarrow{k \to \infty} -\infty.
\end{equation}
This forces the implicit multiplier to vanish,
\begin{equation}
    \gamma_m^* = \lim_{k \to \infty} 
    \frac{\exp\left( \frac{g_m(\bm{\Psi}^{j_k})}{\mu^{j_k}} \right)}
    {1 + \exp\left( \frac{g_m(\bm{\Psi}^{j_k})}{\mu^{j_k}} \right)} = 0.
\end{equation}
Given that $\lambda^*$ is bounded in \eqref{eq:lambda_stabilize}, the effective multiplier converges to zero,
\begin{equation}
    \nu_m^* = \lambda^* \gamma_m^* = 0.
\end{equation}
Thus, $\nu_m^* g_m(\bm{\Psi}^*) = 0$, proving complementary slackness.

\subsection*{3. Stationarity}
Finally, we establish stationarity. Define the KKT residual vector $v \in T_{\bm{\Psi}^*}\mathcal{M}$ at the limit point as,
\begin{equation}
    v \triangleq \operatorname{grad} f_0(\bm{\Psi}^*) + \sum_{m \in \mathcal{C}} \nu_m^* \operatorname{grad} g_m(\bm{\Psi}^*).
\end{equation}
To prove $\|v\|_F = 0$, we introduce the parallel transport to map gradients from $T_{\bm{\Psi}^{j_k}}\mathcal{M}$ to $T_{\bm{\Psi}^*}\mathcal{M}$.
Since $\bm{\Psi}^{j_k} \to \bm{\Psi}^*$, for sufficiently large $k$, the iterate falls within the injectivity radius of the limit point \cite{absil2008optimization},
\begin{equation}
    \operatorname{dist}(\bm{\Psi}^{j_k}, \bm{\Psi}^*) < \operatorname{inj}(\bm{\Psi}^*).
\end{equation}
This condition guarantees the existence of a unique minimizing geodesic connecting $\bm{\Psi}^{j_k}$ and $\bm{\Psi}^*$. Let $\mathcal{P}_{k \to *}$ denote the parallel transport operator along this unique geodesic. Leveraging this uniquely operator, we apply the triangle inequality,
\begin{equation} \label{eq:v_triangle}
    \begin{aligned}
    \|v\|_F \le & \underbrace{\left\| v - \mathcal{P}_{k \to *} \left( \operatorname{grad}\phi(\bm{\Psi}^{j_k}) \right) \right\|_F}_{\text{Term A}} \\
    & + \underbrace{\left\| \mathcal{P}_{k \to *} \left( \operatorname{grad}\phi(\bm{\Psi}^{j_k}) \right) \right\|_F}_{\text{Term B}}.
    \end{aligned}
\end{equation}
We analyze the limiting behavior as $k \to \infty$.

\subsubsection*{Term B }
Since the parallel transport operator $\mathcal{P}_{k \to *}$ is an isometry \cite{absil2008optimization}, it preserves the Frobenius norm,
\begin{equation}
    \text{Term B} 
    = \left\| \mathcal{P}_{k \to *} \left( \operatorname{grad}\phi(\bm{\Psi}^{j_k}) \right) \right\|_F 
    = \|\operatorname{grad}\phi(\bm{\Psi}^{j_k})\|_F.
\end{equation}
By the inner-loop stopping criterion in Algorithm~\ref{alg:1}, the iterate $\bm{\Psi}^{j_k}$ satisfies $\|\operatorname{grad}\phi(\bm{\Psi}^{j_k})\|_F =0$ in Theorem \ref{appendixB}, we have,
\begin{equation}
    \lim_{k\to\infty}\|\operatorname{grad}\phi(\bm{\Psi}^{j_k})\|_F = 0 
     \implies 
    \lim_{k\to\infty}\text{Term B} = 0.
\end{equation}

\subsubsection*{Term A }
Expanding Term A via the linearity of parallel transport,
\begin{equation}
    \begin{aligned}
    &\text{Term A} = \|(\operatorname{grad} f_0(\bm{\Psi}^*) + \sum_{m \in \mathcal{C}} \nu_m^* \operatorname{grad} g_m(\bm{\Psi}^*))  \\
    &\quad  - \mathcal{P}_{k \to *} (\operatorname{grad} f_0(\bm{\Psi}^{j_k}) + \sum_{m \in \mathcal{C}} \nu_m^{j_k} \operatorname{grad} g_m(\bm{\Psi}^{j_k})) \|_F.
    \end{aligned}
\end{equation}
Using the triangle inequality to separate the objective and constraint terms,
\begin{equation} \label{eq:termA_split}
    \begin{aligned}
    &\text{Term A} \le \underbrace{\left\| \operatorname{grad} f_0(\bm{\Psi}^*) - \mathcal{P}_{k \to *} (\operatorname{grad} f_0(\bm{\Psi}^{j_k})) \right\|_F}_{\Delta_f} \\
    &\quad + \sum_{m \in \mathcal{C}} \underbrace{\left\| \nu_m^* \operatorname{grad} g_m(\bm{\Psi}^*) - \nu_m^{j_k} \mathcal{P}_{k \to *} (\operatorname{grad} g_m(\bm{\Psi}^{j_k})) \right\|_F}_{\Delta_{g,m}}.
    \end{aligned}
\end{equation}
Since $f_0$ is continuously differentiable, $\operatorname{grad} f_0$ is a continuous vector field. According to the continuity property of vector fields under parallel transport along minimizing geodesics \cite[Lemma A.2]{liu2020simple},
\begin{equation}
    \lim_{k \to \infty} \Delta_f = 0.
\end{equation}
For $\Delta_{g,m}$, adding and subtracting $\nu_m^* \mathcal{P}_{k \to *} (\operatorname{grad} g_m(\bm{\Psi}^{j_k}))$ decouples the convergence,
\begin{equation}
    \begin{aligned}
    \Delta_{g,m} &= \left\| \nu_m^* (\operatorname{grad} g_m(\bm{\Psi}^*) - \mathcal{P}_{k \to *} \operatorname{grad} g_m(\bm{\Psi}^{j_k})) \right. \\
    &\quad \left. + (\nu_m^* - \nu_m^{j_k}) \mathcal{P}_{k \to *} \operatorname{grad} g_m(\bm{\Psi}^{j_k}) \right\|_F \\
    &\le |\nu_m^*| \cdot \left\| \operatorname{grad} g_m(\bm{\Psi}^*) - \mathcal{P}_{k \to *} \operatorname{grad} g_m(\bm{\Psi}^{j_k}) \right\|_F \\
    &\quad + |\nu_m^* - \nu_m^{j_k}| \cdot \|\operatorname{grad} g_m(\bm{\Psi}^{j_k})\|_F.
    \end{aligned}
\end{equation}
The first term tends to $0$ by the continuity of $\operatorname{grad} g_m$ and \cite[Lemma A.2]{liu2020simple}. The second term tends to $0$ since $|\nu_m^* - \nu_m^{j_k}| \to 0$ by \eqref{eq:nu_converge} and the conclusion in Theorem \ref{appendixB}. Therefore, $\lim_{k \to \infty} \Delta_{g,m} = 0$ for each $m$. Substituting these limits into \eqref{eq:v_triangle} yields $\|v\|_F = 0$, confirming stationarity.

In conclusion, since $\bm{\Psi}^*$ satisfies primal feasibility, dual feasibility, complementary slackness, and stationarity, it is a KKT point of the original problem. This completes the proof. \hfill $\blacksquare$

\bibliographystyle{IEEEtran}
\bibliography{Bibliography}

\vfill

\end{document}